\documentclass[a4paper, 9pt, twosides]{amsart}

\usepackage{lmodern}
\usepackage{amssymb}
\usepackage{amsthm}
\usepackage{amsaddr}
\usepackage{mathtools}
\usepackage{MnSymbol}
\usepackage{graphicx}
\usepackage{caption}
\usepackage{subcaption}
\usepackage{url}

\usepackage{enumitem}
\setitemize{nolistsep}
\setenumerate{nolistsep}
\usepackage{todonotes}
\usepackage{algorithm,algorithmic}

\allowdisplaybreaks

\newtheorem{theorem}{Theorem}

\newtheorem{corollary}{Corollary}

\theoremstyle{definition}
\newtheorem{defn}{Definition}
\newtheorem{example}{Example}

\theoremstyle{remark}
\newtheorem{rem}{Remark}

\theoremstyle{assumption}
\newtheorem{assump}{Assumption}

\theoremstyle{fact}
\newtheorem{fact}{Fact}

\theoremstyle{claim}

\theoremstyle{prob}
\newtheorem{prob}{Problem}

\theoremstyle{algo}

\theoremstyle{experiment}

\numberwithin{equation}{section}

\newcommand{\abs}[1]{\left\lvert{#1}\right\rvert}

\newcommand{\norm}[1]{\left\lVert#1\right\rVert}
\newcommand{\pmat}[1]{\begin{pmatrix}#1\end{pmatrix}}

\newcommand{\R}{\mathbb{R}}
\newcommand{\N}{\mathbb{N}}
\renewcommand{\P}{\mathcal{Q}}
\newcommand{\Sw}{\mathcal{S}}
\newcommand{\A}{\mathcal{A}}
\newcommand{\mb}{\overline{m}}
\newcommand{\toone}{\overset{1}{\to}}
\newcommand{\totwo}{\overset{2}{\to}}
\newcommand{\kb}{\overline{\kappa}}
\newcommand{\cb}{\overline{\chi}}
\newcommand{\tor}{\overset{r}{\to}}

\allowdisplaybreaks

\title[]{On stabilizability of switched linear systems\\under restricted switching}
\author{Atreyee Kundu}
\address{Department of Electrical Engineering,\\Indian Institute of Science Bangalore,\\Bengaluru - 560012, India,\\ E-mail: atreyeek@iisc.ac.in}


\date{\today}

\begin{document}

	\begin{abstract}
         This paper deals with stability of discrete-time switched linear systems whose all subsystems are unstable and the set of admissible switching signals obeys pre-specified restrictions on switches between the subsystems and dwell times on the subsystems. We derive sufficient conditions on the subsystems matrices such that a switched system is globally exponentially stable under a set of purely time-dependent switching signals that obeys the given restrictions. The main apparatuses for our analysis are (matrix) commutation relations between certain products of the subsystems matrices and graph-theoretic arguments.
    \end{abstract}

    \maketitle
\section{Introduction}
\label{s:intro}
\subsection{The setting}
\label{ss:prob_setting}
    A \emph{switched system} has two ingredients --- a family of systems and a switching signal. The \emph{switching signal} selects an \emph{active subsystem} at every instant of time, i.e., the system from the family that is currently being followed \cite[\S1.1.2]{Liberzon2003}. It is well-known that a switched system does not necessarily inherit qualitative properties of its constituent subsystems. For instance, divergent trajectories may be generated by switching appropriately among stable subsystems while a suitably constrained switching signal may ensure stability of a switched system even if all subsystems are unstable (see e.g., \cite[p.\ 19]{Liberzon2003} for examples with two subsystems). Due to such interesting features, stability of switched systems constitutes a key topic in the literature.

    In this paper we focus on discrete-time switched linear systems whose all subsystems are unstable and the set of switching signals obeys pre-specified restrictions on admissible switches between the subsystems and admissible dwell times on the subsystems. Such restrictions are inherent to many engineering applications, see e.g., \cite[Remark 1]{kun2018} for examples. We call a switched system \emph{stabilizable} if there exists an admissible switching signal under which the system is globally exponentially stable. Given the set of admissible switches between the subsystems and the admissible minimum and maximum dwell times on the subsystems, we find sufficient conditions on the subsystems matrices such that a switched system under consideration is stabilizable.
\subsection{Prior works}
\label{ss:prior_works}
    Stabilizability of switched systems under unrestricted switching has attracted considerable research attention in the literature.\footnote{By ``unrestricted switching'', we mean that the switches between the subsystems and the dwell times on the subsystems are unrestricted, unlike the setting of this paper.} Given a set of linear unstable subsystems, the problem of deciding whether (or not) there exists a switching signal that stabilizes the resulting switched system, in general, belongs to the class of NP-hard problems, see \cite{Skafidas1999,Vlassis2014} for results and discussions. On the one hand, necessary and sufficient conditions for this problem is proposed only recently in \cite{Fiacchini2014}. Verifying the condition of \cite{Fiacchini2014} involves checking the containment of a set in the union of other sets, and hence possesses inherent computational complexity. On the other hand, sufficient conditions for determining the existence of a stabilizing switching signal are plenty in the literature. Stability of a switched system in the setting of the so-called min-switching signals \cite{Liberzon2003} is studied in \cite{Geromel2006}. The subsystems matrices are required to satisfy a set of bilinear matrix inequalities (BMIs) called the Lyapunov-Metzler inequalities, for stability under these switching signals. In \cite{Sun2011} the existence of min-switching signals is claimed to be necessary and sufficient for exponential stabilizability of a switched system. Algebraic relations of Lyapunov-Metzler inequalities with classical S-procedure characterization \cite{Lure1944} are explored recently in \cite{Heemels2017}. Generalized versions of both sets of inequalities along with their relations to the classical ones are studied in \cite{Kundu2017, Fiacchini2016}. While the above class of results depends on state-dependent switching signals, in \cite{Fiacchini2016} stability of a switched system is also addressed under purely time-dependent periodic switching signals based on the satisfaction of a set of linear matrix inequalities (LMIs). The proposed LMI conditions are equivalent to the generalized versions of Lyapunov-Metzler inequalities, and is implied by the classical Lyapunov-Metzler inequalities \cite{Fiacchini2016}.

    Stabilization of a switched linear system under a set of switching signals restricted by the language of a non-deterministic finite automaton \cite{Hopcroft} is addressed recently in \cite{Fiacchini2018}. Such an automaton captures a large class of constraints on the switching signals. The authors present an algorithm to design stabilizing state-feedback switching signals. It is shown using geometry of certain sets that the termination of this algorithm is a necessary and sufficient condition for recurrent stabilizability, which in turn is a sufficient condition for stabilizability of a switched system.
\subsection{Our contributions}
\label{ss:contri}
    In this paper our main contributions are twofold:
    \begin{itemize}[label = \(\circ\), leftmargin = *]
        \item first, we employ purely time-dependent switching signals that are not restricted to periodic constructions in the study of stabilizability of switched linear systems, and
        \item second, we transcend beyond the regime of Lyapunov functions and/or matrix inequalities based stabilizability conditions and rely on properties of Euclidean norms of (matrix) commutators (Lie brackets) of certain products of the subsystems matrices for this purpose.
    \end{itemize}

    We assume that there exist two subsystems that form a Schur stable combination including values from the admissible dwell time interval. If switches between these two subsystems are unrestricted, then the switched system under consideration admits a stabilizing periodic switching signal. We address the general setting where switches between the two subsystems forming a stable combination may be restricted. Towards this end, we associate a directed graph to the switched system: the vertices of this graph are the indices of the subsystems and the directed edges correspond to the admissible switches between the subsystems. Clearly, a path from a source vertex to a destination vertex on this graph implies that the set of admissible switches allows us to reach the destination subsystem from the source subsystem. We show that if the underlying directed graph of a switched system admits paths between these subsystems that satisfy certain conditions involving the following components: (i) the rate of decay of the Schur stable combination, (ii) upper bounds on the Euclidean norms of the (matrix) commutators of certain products of the subsystems matrices that appear in the paths, and (iii) certain scalars capturing the properties of the subsystems matrices, then the switched system under consideration is stabilizable under restricted switching. In case of multiple favourable paths, our construction of stabilizing switching signals is non-periodic, while with unique choice of these paths, we construct stabilizing periodic switching signals. In both cases, the stabilizing switching signals activate the subsystems that appear in the paths under consideration, in different patterns, and dwell on them for admissible durations of time. We demonstrate the settings where switches between the subsystems forming a Schur stable combination are partially restricted or unrestricted as special cases of our results.

    The primary features of our results are the following:
    \begin{itemize}[label = \(\circ\), leftmargin = *]
        \item Our conditions for stabilizability involve scalar inequalities, and are, therefore, numerically easier to verify compared to the existing matrix inequalities based conditions.
        \item We rely on upper bounds on the Euclidean norm of the matrix commutators of certain products of the subsystems matrices, thereby adding an inherent robustness to our stabilizability conditions in the sense that if the elements of the subsystems matrices are only partially known and/or they evolve over time in a manner such that the matrices are not ``too far'' from a set of matrices for which the products under consideration commute, then the switched system continues to be stable under our switching signals.
    \end{itemize}
    Matrix commutators (Lie brackets) have been employed to study stability of switched systems with all or some stable subsystems earlier in the literature, see e.g., \cite{Narendra1994, Agrachev2012, stab_cycle, stab_min_dwell, stab_min_max_dwell}. However, to the best of our knowledge, this is the first instance where a blend of directed graphs and matrix commutators is employed to address stability of switched systems with all unstable subsystems under restricted switching, see Remark \ref{rem:compa3} for a comparative discussion of our results with the existing matrix commutator based stability conditions for switched systems.
\subsection{Paper organization}
\label{ss:paper_org}
     The remainder of this paper is organized as follows: in \S\ref{s:prob_stat} we formulate the problem under consideration. A set of preliminaries for our results are described in \S\ref{s:prelims}. Our results appear in \S\ref{s:res}. We present a numerical example in \S\ref{s:numex} and conclude in \S\ref{s:concln} with a brief discussion of future research directions.
\subsection{Notation}
\label{ss:notation}
     \(\R\) is the set of real numbers and \(\N\) is the set of natural numbers, \(\N_{0} = \N\cup\{0\}\). \(\norm{\cdot}\) denotes the Euclidean norm (resp., induced matrix norm) of a vector (resp., a matrix). For a matrix \(\mathcal{M}\), given by a product of matrices \(M_{k}\)’s, \(\abs{\mathcal{M}}\) denotes the length of the product, i.e., the number of matrices that appear in \({\mathcal{M}}\), counting repetitions.
\section{Problem statement}
\label{s:prob_stat}
    We consider a family of systems
        \begin{align}
        \label{e:family}
            x(t+1) = A_{\ell}x(t),\:\:x(0) = x_{0},\:\:\ell\in\P,\:\:t\in\N_{0},
        \end{align}
    where \(x(t)\in\R^{d}\) is the vector of states at time \(t\), \(\P = \{1,2,\ldots,N\}\) is an index set, and \(A_{\ell}\in\R^{d\times d}\), \(\ell\in\P\) are constant unstable matrices.\footnote{A matrix \(\mathcal{M}\in\R^{d\times d}\) is Schur stable if all its eigenvalues are inside the open unit disk. \(\mathcal{M}\) is unstable if it is not Schur stable.} Let \(\sigma:\N_{0}\to\P\) be a switching signal that specifies at every time \(t\), the index of the active subsystem, i.e., the dynamics from \eqref{e:family} that is being followed at \(t\). A switched system \cite[\S 1.1.2]{Liberzon2003} generated by the family of systems \eqref{e:family} and a switching signal \(\sigma\) is given by
    \begin{align}
    \label{e:swsys}
        x(t+1) = A_{\sigma(t)}x(t),\:\:x(0)=x_{0},\:\:t\in\N_{0}.
    \end{align}

    Let \(E(\P)\) be the set of all ordered pairs \((k,\ell)\) such that a switch from subsystem \(k\) to subsystem \(\ell\) is admissible, \(k,\ell\in\P\), \(k\neq \ell\), and \(\delta\) and \(\Delta\) denote the admissible minimum and maximum dwell times on each subsystem \(\ell\in\P\), respectively, \(0 < \delta < \Delta\). Let \(0=:\tau_{0}<\tau_{1}<\cdots\) be the \emph{switching instants}; these are the points in time where \(\sigma\) jumps. We call a switching signal \(\sigma\) \emph{admissible} if it satisfies the following two conditions: \((\sigma(\tau_{h}),\sigma(\tau_{h+1}))\in E(\P)\) and \(\delta\leq\tau_{h+1}-\tau_{h}\leq\Delta\), \(h=0,1,2,\ldots\). Let \(\Sw\) denote the set of all admissible switching signals. The solution to \eqref{e:swsys} is given by
    \[
        x(t) = A_{\sigma(t-1)}\ldots A_{\sigma(1)}A_{\sigma(0)}x_{0},\:\:t\in\N,
    \]
    where we have suppressed the dependence of \(x\) on \(\sigma\in\Sw\) for notational simplicity.

    It is evident that the existence of a \(\sigma\in\Sw\) that ensures stability of \eqref{e:swsys} depends on the following components: (a) the unstable subsystems matrices, \(A_{\ell}\), \(\ell\in\P\), (b) the set of admissible switches, \(E(\P)\), and (c) the admissible dwell times on the subsystems, \(\delta\) and \(\Delta\). We will consider (b) and (c) to be ``given'', and our objective is to ``identify'' conditions on (a) such that \(\Sw\) admits a stabilizing switching signal. Recall that

    \begin{defn}{\cite[\S2]{Agrachev2012}}
    \label{d:gues}
    \rm{
        The switched system \eqref{e:swsys} is \emph{globally exponentially stable (GES) under a switching signal \(\sigma\in\Sw\)} if there exist positive numbers \(c\) and \(\lambda\) such that for arbitrary choices of the initial condition \(x_{0}\), the following inequality holds:
        \begin{align}
        \label{e:ges1}
            \norm{x(t)}\leq ce^{-\lambda t}\norm{x_{0}}\:\:\text{for all}\:t\in\N,
        \end{align}
        where \(\norm{v}\) denotes the Euclidean norm of a vector \(v\).
    }
    \end{defn}
    \begin{defn}
    \label{d:stabilizability}
    \rm{
        The switched system \eqref{e:swsys} is called \emph{stabilizable} if there exists a switching signal \(\sigma\in\Sw\) under which \eqref{e:swsys} is
        GES.
    }
    \end{defn}
    We will solve the following problem:
    \begin{prob}
    \label{prob:mainprob}
    \rm{
         Given the set of admissible switches between the subsystems, \(E(\P)\), and the admissible minimum and maximum dwell times on the subsystems, \(\delta\) and \(\Delta\), find conditions on the matrices \(A_{\ell}\), \(\ell\in\P\), such that the switched system \eqref{e:swsys} is stabilizable.
    }
    \end{prob}

    Notice that we are seeking for a \(\sigma\) that obeys the pre-specified restrictions on admissible switches between the subsystems and admissible dwell times on the subsystems \emph{and} ensures stability of \eqref{e:swsys}. At this point, it is worth highlighting that Problem \ref{prob:mainprob} does not admit a trivial solution even when \eqref{e:family} admits stable subsystems. Indeed, in the presence of restrictions on admissible maximum dwell times, a constant switching signal on a stable subsystem is not admissible, and switching between stable subsystems does not necessarily preserve stability of a switched system.

    Prior to presenting our solution to Problem \ref{prob:mainprob} we catalog a set of preliminaries.
\section{Preliminaries}
\label{s:prelims}
    We will consider that
     \begin{assump}
    \label{a:stable_combi}
    \rm{
        There exist \(i,j\in\P\) and \(p,q\in\{\delta,\delta+1,\ldots,\Delta\}\) such that the matrix \(\A := A_{i}^{p}A_{j}^{q}\) is Schur stable.
    }
    \end{assump}

    It follows from the properties of Schur stable matrices that
   \begin{fact}
   \label{fact:m_defn}
   \rm{
        There exists \(m\in\N\) such that the following condition holds:
        \begin{align}
        \label{e:m_ineq}
            \norm{\A^{m}}\leq\rho < 1.
        \end{align}
   }
   \end{fact}

   The motivation behind Assumption \ref{a:stable_combi} is the use of purely time-dependent switching signals for the stabilization of \eqref{e:swsys}. Since we do not utilize (or assume access to) information about system state, \(x(t)\), \(t\in\N_{0}\), the existence of a stable combination formed by two subsystems matrices is useful in our analysis. Let switches between the subsystems \(i\) and \(j\) be unrestricted, i.e., \((i,j)\) and \((j,i)\in E(\P)\). Consider a switching signal \(\sigma\) that activates the subsystems \(j\) and \(i\) alternatively and dwells on them for \(q\) and \(p\) units of time, respectively. Mathematically speaking, \(\sigma\) satisfies
        \(
            \sigma(\tau_{0}) = j,\:\:\sigma(\tau_{i}) = i,\:\:\sigma(\tau_{2}) = j,\:\:\sigma(\tau_{3}) = i,\ldots
        \)
        with
        \begin{align*}
            \tau_{h+1} - \tau_{h} =
            \begin{cases}
                q,\:\:\text{if}&\:\sigma(\tau_{h}) = j,\\
                p,\:\:\text{if}&\:\sigma(\tau_{h}) = i,
            \end{cases}
            \:\:h=0,1,2,\ldots.
        \end{align*}
        Since \((j,i)\) and \((i,j)\in E(\P)\) and \(p,q\in\{\delta,\delta+1,\ldots,\Delta\}\), it follows that \(\sigma\in\Sw\). GES of \eqref{e:swsys} under \(\sigma\) is immediate from Fact \ref{fact:m_defn}. Indeed, we are dealing with stability of the difference equation
        \(
            x(\overline{t}+1) = \A x(\overline{t}),\:\:x(0) = x_{0},
        \)
        where \(\A\) is Schur stable and \(\overline{t}:t = (p+q):1\).

        However, since admissible switches between the subsystems are restricted, Assumption \ref{a:stable_combi} does not lead to a trivial solution to Problem \ref{prob:mainprob} if either \((i,j)\) or \((j,i)\) or both \((i,j)\) and \((j,i)\) are not elements of \(E(\P)\). In the sequel we aim for a general solution to Problem \ref{prob:mainprob} that works as long as it is possible to reach subsystem \(i\) from subsystem \(j\) and vice-versa, through a set of favourable subsystems. Clearly, we need to take into account the properties of the subsystems matrices that appear between \(i\) and \(j\) in this case. We shall identify the settings where the switches between the subsystems \(i\) and \(j\) are partially restricted (either \((i,j)\) or \((j,i)\in E(\P)\)) and unrestricted (both \((i,j)\) and \((j,i)\in E(\P)\)) as special cases of our results.

   Given \(m\in\N\), let \(\mb\in\N\) be such that
   \begin{align}
   \label{e:mbar_defn}
        \overline{m} =
        \begin{cases}
            1,\:\:&\text{if}\:m\in\{1,2\},\\
            2,\:\:&\text{if}\:m\in\{3,4,5\},\\
            3,\:\:&\text{if}\:m\in\{6,7,8,9\},\\
            4,\:\:&\text{if}\:m\in\{10,11,12,13,14\},\\
            \vdots
        \end{cases}
   \end{align}
   Let
    \begin{align}
    \label{e:M_defn}
        M := \max_{\ell\in\P}\norm{A_{\ell}}.
    \end{align}

    We associate a directed graph \(G(\P,E(\P))\) with the family of systems \eqref{e:family} in the following manner:
    \begin{itemize}[label = \(\circ\), leftmargin = *]
        \item the set of vertices of \(G\) is the index set, \(\P\), and
        \item the set of edges of \(G\) is the set of admissible switches, \(E(\P)\).
    \end{itemize}

    \begin{defn}
    \label{d:path}
    \rm{
        Fix \(u,v\in\P\). A \emph{\(u\to v\) path on \(G(\P,E(\P))\)} is a finite alternating sequence, \(P_{G}^{u\to v} = w_{0},(w_{0},w_{1}),w_{1},\ldots\),\\\(w_{n-1}\),\((w_{n-1},w_{n}),w_{n}\), where 
        \begin{itemize}[label = \(\circ\), leftmargin = *]
            \item \(w_{k}\in\P\), \(k=0,1,\ldots,n\),
            \item \((w_{k},w_{k+1})\in E(\P)\), \(k=0,1,\ldots,n-1\),
            \item \(w_{0} = u\),
            \item \(w_{n} = v\),
            \item \(w_{k}\neq u\), \(w_{k}\neq v\), \(k=1,2,\ldots,n-1\).
        \end{itemize}
    }
    \end{defn}

    \begin{defn}
    \label{d:length}
    \rm{
        The \emph{length of a \(u\to v\) path on \(G(\P,E(\P))\)}, denoted by \(\abs{P_{G}^{u\to v}}\), is the total number of vertices that appear in \(P_{G}^{u\to v}\) excluding the initial vertex \(u\) and the final vertex \(v\). For example, the length of \(P_{G}^{u\to v}\), described in Definition \ref{d:path}, is \(n-1\).
    }
    \end{defn}

    \begin{defn}
    \label{d:comm}
    \rm{
        For a \(u\to v\) path on \(G(\P,E(\P))\), we define the following set of (matrix) commutators of products of matrices:
        \begin{align}
        \label{e:comm}
            F_{u\to v,\ell}^{a,b} := A_{\ell}^{a}(A_{w_{n-1}}^{b}\ldots A_{w_{1}}^{b})-(A_{w_{n-1}}^{b}\ldots A_{w_{1}}^{b})A_{\ell}^{a},
        \end{align}
        where \(\ell\in\P\), \(w_{k}\neq\ell\), \(k=1,2,\ldots,n-1\), and \(a,b\in\N\).
    }
    \end{defn}
    We note that for a \(u\to v\) path, \(P_{G}^{u\to v} = u,(u,v),v\), on \(G(\P,E(\P))\), \(\abs{P_{G}^{u\to v}} = 0\) and \(F_{u\to v,\ell}^{a,b} = 0\) for any \(\ell\in\P\) and \(a,b\in\N\).

    We are now in a position to present our solutions to Problem \ref{prob:mainprob}.
\section{Results}
\label{s:res}
    Let \(i,j\in\P\) satisfy Assumption \ref{a:stable_combi}. Fix two pairs of \(j\to i\) and \(i\to j\) paths, \((P_{G}^{j\toone i},P_{G}^{i\toone j})\) and \((P_{G}^{j\totwo i},P_{G}^{i\totwo j})\) on \(G(\P,E(\P))\). We define the following quantities
    \begin{align*}
            \xi_{j\overset{1}{\to} i,i} &= \abs{P_{G}^{j\overset{1}{\to} i}}\delta(\mb-1)
            +\abs{P_{G}^{i\overset{1}{\to} j}}\delta\overline{m} +
            \biggl(\abs{P_{G}^{j\totwo i}}+\abs{P_{G}^{i\totwo j}}\biggr)\delta(m-\mb)+ p(m-1) + qm,\nonumber\\
            \xi_{j\overset{2}{\to} i,i} &= \biggl(\abs{P_{G}^{j\overset{1}{\to} i}}+\abs{P_{G}^{i\overset{1}{\to} j}}\biggr)\delta\mb + \abs{P_{G}^{j\overset{2}{\to} i}}\delta(m-\mb-1)+\abs{P_{G}^{i\overset{2}{\to} j}}\delta(m-\mb)+p(m-1)+qm,\nonumber\\
            \xi_{j\toone i,j} &= \abs{P_{G}^{j\toone i}}\delta(\mb-1)+\abs{P_{G}^{i\toone j}}\delta\mb+\biggl(\abs{P_{G}^{j \totwo i}}+\abs{P_{G}^{i\totwo j}}\biggr)\delta(m-\mb)+pm+q(m-1),\nonumber\\
            \xi_{j\totwo i,j} &= \biggl(\abs{P_{G}^{j\toone i}}+\abs{P_{G}^{i\toone j}}\biggr)\delta\mb + \abs{P_{G}^{j\totwo i}}\delta(m-\mb-1)+\abs{P_{G}^{i\totwo j}}\delta(m-\mb)+pm+q(m-1),\nonumber\\
            \xi_{i\toone j,i} &= \abs{P_{G}^{j\toone i}}\delta\mb + \abs{P_{G}^{i\toone j}}\delta(\mb-1)+\biggl(\abs{P_{G}^{j\totwo i}}+\abs{P_{G}^{i\totwo j}}\biggr)\delta(m-\mb)+ p(m-1) + qm,\\
            \xi_{i\totwo j,i} &= \biggl(\abs{P_{G}^{j\toone i}}+\abs{P_{G}^{i\toone j}}\biggr)\delta\mb+\abs{P_{G}^{j\totwo i}}\delta(m-\mb)+\abs{P_{G}^{i\totwo j}}\delta(m-\mb-1)+p(m-1)+qm,\nonumber\\
            \xi_{i\toone j,j} &= \abs{P_{G}^{j\toone i}}\delta\mb + \abs{P_{G}^{i\toone j}}\delta(\mb-1) + \biggl(\abs{P_{G}^{j\totwo i}}+\abs{P_{G}^{i\totwo j}}\biggr)\delta(m-\mb)+pm+q(m-1),\nonumber\\
            \xi_{i\totwo j,j} &= \biggl(\abs{P_{G}^{j\toone i}}+\abs{P_{G}^{i\toone j}}\biggr)\delta\mb + \abs{P_{G}^{j\totwo i}}\delta(m-\mb) + \abs{P_{G}^{i\totwo j}}\delta(m-\mb-1)+pm+q(m-1),\nonumber\\
            \xi_{j\toone i\toone j} &= \Biggl(\biggl(\abs{P_{G}^{j\toone i}}+\abs{P_{G}^{i\toone j}}\biggr)\delta+p+q\Biggr)\mb,\nonumber\\
             \xi_{j\totwo i\totwo j} &= \Biggl(\biggl(\abs{P_{G}^{j\totwo i}}+\abs{P_{G}^{i\totwo j}}\biggr)\delta+p+q\Biggr)(m-\mb).\nonumber
        \end{align*}
         \begin{theorem}
   \label{t:mainres1}
        Let \(i,j\in\P\) satisfy Assumption \ref{a:stable_combi} and \(\lambda\) be an arbitrary positive number satisfying
        \begin{align}
        \label{e:maincondn1}
            \rho e^{\lambda m} < 1.
        \end{align}
        Suppose that \(G(\P,E(\P))\) admits two pairs of \(j\to i\) and \(i\to j\) paths, \((P_{G}^{j\overset{r}{\to} i},P_{G}^{i\overset{r}{\to} j})\), \(r=1,2\), that satisfy the following condition: there exist scalars \(\varepsilon_{j\overset{r}{\to} i,i}\), \(\varepsilon_{j\overset{r}{\to} i,j}\), \(\varepsilon_{i\overset{r}{\to} j,i}\), \(\varepsilon_{i\overset{r}{\to} j,j}\), \(r=1,2\), small enough, such that
        \begin{align}
        \label{e:maincondn2}
            \norm{F_{j\overset{r}{\to} i,i}^{p,\delta}}\leq\varepsilon_{j\overset{r}{\to}i,i},
             \norm{F_{j\overset{r}{\to} i,j}^{q,\delta}}\leq\varepsilon_{j\overset{r}{\to}i,j},\:\:r=1,2,
        \end{align}
        \begin{align}
        \label{e:maincondn3}
            \norm{F_{i\overset{r}{\to} j,i}^{p,\delta}}\leq\varepsilon_{i\overset{r}{\to}j,i},
             \norm{F_{i\overset{r}{\to} j,j}^{q,\delta}}\leq\varepsilon_{i\overset{r}{\to}j,j},\:\:r=1,2,
        \end{align}
        and
        \begin{align}
        \label{e:maincondn4}
            \rho e^{\lambda m} + \Biggl(&\frac{\mb(\mb-1)}{2}M^{\xi_{j\overset{1}{\to} i,i}}\varepsilon_{j\overset{1}{\to} i,i}
            + \frac{\mb(\mb+1)}{2}\biggl(M^{\xi_{j\overset{1}{\to} i,j}}\varepsilon_{j\overset{1}{\to} i,j}+M^{\xi_{i\overset{1}{\to} j,i}}\varepsilon_{i\overset{1}{\to} j,i}+M^{\xi_{i\overset{1}{\to} j,j}}\varepsilon_{i\overset{1}{\to} j,j}\biggr)\nonumber\\
            &\:\:+\frac{m(m-1)-\mb(\mb-1)}{2}M^{\xi_{j\overset{2}{\to} i,i}}\varepsilon_{j\overset{2}{\to} i,i}
            +\frac{m(m+1)-\mb(\mb+1)}{2}\biggl(M^{\xi_{j\overset{2}{\to} i,j}}\varepsilon_{j\overset{2}{\to} i,j}\nonumber\\
            &\quad\quad+M^{\xi_{i\overset{2}{\to} j,i}}\varepsilon_{i\overset{2}{\to} j,i}+M^{\xi_{i\overset{2}{\to} j,j}}\varepsilon_{i\overset{2}{\to} j,j}\biggr)\Biggr)
            \times e^{\lambda\biggl(\xi_{j\overset{1}{\to}i\overset{1}{\to}j}+\xi_{j\overset{2}{\to}i\overset{2}{\to}j}\biggr)}\leq 1.
        \end{align}
        Then there exists a non-periodic switching signal \(\sigma\in\Sw\) under which the switched system \eqref{e:swsys} is GES.
   \end{theorem}

   \begin{proof}
        Let
        \begin{align*}
            P_{G}^{j\toone i} &= w_{0}^{(1)},(w_{0}^{(1)},w_{1}^{(1)}),w_{1}^{(1)},\ldots,w_{\ell_{1}-1}^{(1)},
            (w_{\ell_{1}-1}^{(1)},w_{\ell_{1}}^{(1)}),w_{\ell_{1}}^{(1)},\\
            P_{G}^{j\totwo i} &= w_{0}^{(2)},(w_{0}^{(2)},w_{1}^{(2)}),w_{1}^{(2)},\ldots,w_{\ell_{2}-1}^{(2)},
            (w_{\ell_{2}-1}^{(2)},w_{\ell_{2}}^{(2)}),w_{\ell_{2}}^{(2)},\\
            P_{G}^{i\toone j} &= w_{\ell_{1}}^{(1)},(w_{\ell_{1}}^{(1)},w_{\ell_{1}+1}^{(1)}),w_{\ell_{1}+1}^{(1)},\ldots,w_{n_{1}-1}^{(1)},
            (w_{n_{1}-1}^{(1)},w_{n_{1}}^{(1)}),w_{n_{1}}^{(1)},\\
            P_{G}^{i\totwo j} &= w_{\ell_{2}}^{(2)},(w_{\ell_{2}}^{(2)},w_{\ell_{2}+1}^{(2)}),w_{\ell_{2}+1}^{(2)},\ldots,w_{n_{2}-1}^{(2)},
            (w_{n_{2}-1}^{(2)},w_{n_{2}}^{(2)}),w_{n_{2}}^{(2)}
        \end{align*}
        satisfy conditions \eqref{e:maincondn2}-\eqref{e:maincondn4}.

        Consider a switching signal \(\sigma\) that activates the sequence of subsystems \(w_{0}^{(1)},w_{1}^{(1)},\ldots\),\(w_{\ell_{1}}^{(1)}\),\\\(w_{\ell_{1}+1}^{(1)},\ldots,w_{n_{1}-1}^{(1)}\) followed by \(s\)-many instances of the sequence of subsystems
        \(w_{0}^{(2)},w_{1}^{(2)}\),\(\ldots,w_{\ell_{2}}^{(2)}\),\\\(w_{\ell_{2}+1}^{(2)},\ldots,w_{n_{2}-1}^{(2)}\) repeatedly, \(s=1,2,3,\ldots\) with dwell times \(p\), \(q\) and \(\delta\) units of time on the subsystems \(i\), \(j\) and \(k\in\P\setminus\{i,j\}\), respectively. More precisely, consider the following sets of assignments:
        \begin{itemize}[label = \(\circ\), leftmargin = *]
            \item for \(h=c_{1}n_{1}+c_{2}n_{2},c_{1}n_{1}+c_{2}n_{2}+1,\ldots,(c_{1}+1)n_{1}+c_{2}n_{2}-1\)
            \begin{align}
            \label{e:assign1}
                \sigma(\tau_{h}) = w_{h}^{(1)} - (c_{1}n_{1}+c_{2}n_{2}),
            \end{align}
            and
            \item for \(h=(c_{1}+1)n_{1}+c_{2}n_{2},(c_{1}+1)n_{1}+c_{2}n_{2}+1,\ldots,\bigl((c_{1}+1)n_{1}+(c_{2}+1)n_{2}\bigr)-1\)
            \begin{align}
            \label{e:assign2}
                \sigma(\tau_{h}) = w_{h}^{(2)} - \bigl((c_{1}+1)n_{1}+c_{2}n_{2}\bigr),
            \end{align}
            with
            \begin{align*}
                \tau_{h+1} - \tau_{h} =
                \begin{cases}
                    p,\:\:&\text{if}\:\sigma(\tau_{h}) = i,\\
                    q,\:\:&\text{if}\:\sigma(\tau_{h}) = j,\\
                    \delta,\:\:&\text{if}\:\sigma(\tau_{h}) \in \P\setminus\{i,j\},
                \end{cases}
                h=0,1,2,\ldots.
            \end{align*}
        \end{itemize}
        The switching signal \(\sigma\) under consideration can be constructed as follows:
        \begin{itemize}[label = \(\circ\), leftmargin = *]
            \item set \(\tau_{0} = 0\),
            \item set \(c_{1} = 0\), \(c_{2} = 0\), apply \eqref{e:assign1}-\eqref{e:assign2},
            \item set \(c_{1} = 1\), \(c_{2} = 1\), apply \eqref{e:assign1}-\eqref{e:assign2},
            \item set \(c_{2} = 2\), apply \eqref{e:assign2},
            \item set \(c_{1} = 2\), \(c_{2} = 3\), apply \eqref{e:assign1}-\eqref{e:assign2},
            \item set \(c_{2} = 4\), apply \eqref{e:assign2},
            \item set \(c_{2} = 5\), apply \eqref{e:assign2},
            \item set \(c_{1} = 3\), \(c_{2} = 6\), apply \eqref{e:assign1}-\eqref{e:assign2},
            \item \(\vdots\)
        \end{itemize}

        We first show that \(\sigma\in\Sw\). By construction of \(\sigma\), its values at two consecutive switching instants \(\tau_{h}\) and \(\tau_{h+1}\) are two vertices of \(G(\P,E(\P))\) such that there is a directed edge from vertex (subsystem) \(\sigma(\tau_{h})\) to vertex (subsystem) \(\sigma(\tau_{h+1})\), \(h=0,1,2,\ldots\). Moreover, \(\tau_{h+1}-\tau_{h}\in\{\delta,p,q\}\subseteq\{\delta,\delta+1,\ldots,\Delta\}\). Thus, \(\sigma\in\Sw\). Also, by construction, \(\sigma\) is non-periodic.

        We next show that \eqref{e:swsys} is GES under \(\sigma\). Let \(W\) be the matrix product corresponding to \(\sigma\), defined as \(W = \cdots A_{\sigma(2)}A_{\sigma(1)}A_{\sigma(0)}\). We let \(\overline{W}\) denote an initial segment of \(W\). The length of \(\overline{W}\) is denoted by \(\abs{\overline{W}}\). Then the condition \eqref{e:ges1} for GES of \eqref{e:swsys} under \(\sigma\) can be written equivalently as \cite[\S2]{Agrachev2012}: for every initial segment \(\overline{W}\) of \(W\), we have
        \begin{align}
        \label{e:ges2}
            \norm{\overline{W}}\leq ce^{-\lambda\abs{\overline{W}}}.
        \end{align}
        We apply mathematical induction on the length of an initial segment \(\overline{W}\) to establish \eqref{e:ges2}.

        {\it A. Induction basis}: Pick \(c\) large enough so that \eqref{e:ges2} holds for \(\overline{W}\) satisfying \(\abs{\overline{W}}\geq \xi_{j\toone i\toone j}+\xi_{j\totwo i\totwo j}\).

        {\it B. Induction hypothesis}: Let \(\abs{\overline{W}}\geq \xi_{j\toone i\toone j}+\xi_{j\totwo i\totwo j}+1\) and assume that \eqref{e:ges2} is proved for all products of length less than \(\abs{\overline{W}}\).

        {\it C. Induction step}: Let \(\overline{W} = LR\), where \(\abs{L} = \xi_{j\toone i\toone j}+\xi_{j\totwo i\totwo j}\). It follows by construction of \(\sigma\) that \(L\) contains exactly \(m\)-many \(A_{i}^{p}\) and \(m\)-many \(A_{j}^{q}\). Let us rewrite \(L\) as \(L=\A^{m}L_{1}+L_{2}\), where
        \[
            \abs{L_{1}} = \biggl(\abs{P_{G}^{j\toone i}}+\abs{P_{G}^{i\toone j}}\biggr)\delta\mb + \biggl(\abs{P_{G}^{j\totwo i}}+\abs{P_{G}^{i\totwo j}}\biggr)\delta(m-\mb)
        \]
        and \(L_{2}\) contains the following:
        \begin{itemize}[label = \(\circ\), leftmargin = *]
            \item \(\frac{\mb(\mb-1)}{2}\)-many terms of length \(\xi_{j\toone i,i}+1\) with \(\xi_{j\toone i,i}\)-many \(A_{r}\), \(r\in\P\) and \(1\) \(F_{j\toone i,i}^{p,\delta}\),
            \item \(\frac{m(m-1)-\mb(\mb-1)}{2}\)-many terms of length \(\xi_{j\totwo i,i}+1\) with \(\xi_{j\totwo i,i}\)-many \(A_{r}\), \(r\in\P\) and \(1\) \(F_{j\totwo i,i}^{p,\delta}\),
            \item \(\frac{\mb(\mb+1)}{2}\)-many terms of length \(\xi_{j\toone i,j}+1\) with \(\xi_{j\toone i,j}\)-many \(A_{r}\), \(r\in\P\) and \(1\) \(F_{j\toone i,j}^{q,\delta}\),
            \item \(\frac{m(m+1)-\mb(\mb+1)}{2}\)-many terms of length \(\xi_{j\totwo i,j}+1\) with \(\xi_{j\totwo i,j}\)-many \(A_{r}\), \(r\in\P\) and \(1\) \(F_{j\totwo i,j}^{q,\delta}\),
            \item \(\frac{\mb(\mb+1)}{2}\)-many terms of length \(\xi_{i\toone j,i}+1\) with \(\xi_{i\toone j,i}\)-many \(A_{r}\), \(r\in\P\) and \(1\) \(F_{i\toone j,i}^{p,\delta}\),
            \item \(\frac{m(m+1)-\mb(\mb+1)}{2}\)-many terms of length \(\xi_{i\totwo j,i}+1\) with \(\xi_{i\totwo j,i}\)-many \(A_{r}\), \(r\in\P\) and \(1\) \(F_{i\totwo j,i}^{p,\delta}\),
            \item \(\frac{\mb(\mb+1)}{2}\)-many terms of length \(\xi_{i\toone j,j}+1\) with \(\xi_{i\toone j,j}\)-many \(A_{r}\), \(r\in\P\) and \(1\) \(F_{i\toone j,j}^{q,\delta}\),
            \item \(\frac{m(m+1)-\mb(\mb+1)}{2}\)-many terms of length \(\xi_{i\totwo j,j}+1\) with \(\xi_{i\totwo j,j}\)-many \(A_{r}\), \(r\in\P\) and \(1\) \(F_{i\totwo j,j}^{q,\delta}\).
        \end{itemize}

        (Consider, for example,
        \begin{align*}
            P_{G}^{j\toone i} &= j,(j,3),3,(3,i),i,\:\:
            P_{G}^{j\totwo i} = j,(j,4),4,(4,5),5,(5,i),i,\\
            P_{G}^{i\toone j} &= i,(i,1),1,(1,2),2,(2,j),j,\:\:
            P_{G}^{i\totwo j} = i,(i,6),6,(6,j),j.
        \end{align*}
        Suppose that \(m=2\), \(\delta=p=q = 2\). Let \(L = A_{6}^{2}A_{i}^{2}A_{5}^{2}A_{4}^{2}A_{j}^{2}A_{2}^{2}A_{1}^{2}A_{i}^{2}A_{3}^{2}A_{j}^{2}\). It can be rewritten as
        \begin{align*}
            L &= \underline{A_{6}^{2}A_{i}^{2}}A_{5}^{2}A_{4}^{2}A_{j}^{2}A_{2}^{2}A_{1}^{2}A_{i}^{2}A_{3}^{2}A_{j}^{2}\\ &=A_{i}^{2}A_{6}^{2}\underline{A_{5}^{2}A_{4}^{2}A_{j}^{2}}A_{2}^{2}A_{1}^{2}A_{i}^{2}A_{3}^{2}A_{j}^{2} - F_{i\totwo j,i}^{2,2}A_{5}^{2}A_{4}^{2}A_{j}^{2}A_{2}^{2}A_{1}^{2}A_{i}^{2}A_{3}^{2}A_{j}^{2}\\
            &=A_{i}^{2}\underline{A_{6}^{2}A_{j}^{2}}A_{5}^{2}A_{4}^{2}A_{2}^{2}A_{1}^{2}A_{i}^{2}A_{3}^{2}A_{j}^{2}
            -A_{i}^{2}A_{6}^{2}F_{j\totwo i,j}^{2,2}A_{2}^{2}A_{1}^{2}A_{i}^{2}A_{3}^{2}A_{j}^{2}
            -F_{i\totwo j,i}A_{5}^{2}A_{4}^{2}A_{j}^{2}A_{2}^{2}A_{1}^{2}A_{i}^{2}A_{3}^{2}A_{j}^{2}\\
            &=A_{i}^{2}A_{j}^{2}A_{6}^{2}A_{5}^{2}A_{4}^{2}\underline{A_{2}^{2}A_{1}^{2}A_{i}^{2}}A_{3}^{2}A_{j}^{2}
            -A_{i}^{2}F_{i\totwo j,j}^{2,2}A_{5}^{2}A_{4}^{2}A_{2}^{2}A_{1}^{2}A_{i}^{2}A_{3}^{2}A_{j}^{2}
            -A_{i}^{2}A_{6}^{2}F_{j\totwo i,j}^{2,2}A_{2}^{2}A_{1}^{2}A_{i}^{2}A_{3}^{2}A_{j}^{2}\\
            &\:\:-F_{i\totwo j,i}A_{5}^{2}A_{4}^{2}A_{j}^{2}A_{2}^{2}A_{1}^{2}A_{i}^{2}A_{3}^{2}A_{j}^{2}\\
            &=A_{i}^{2}A_{j}^{2}A_{6}^{2}\underline{A_{5}^{2}A_{4}^{2}A_{i}^{2}}A_{2}^{2}A_{1}^{2}A_{3}^{2}A_{j}^{2}
            -A_{i}^{2}A_{j}^{2}A_{6}^{2}A_{5}^{2}A_{4}^{2}F_{i\toone j,i}^{2,2}A_{3}^{2}A_{j}^{2}
            -A_{i}^{2}F_{i\totwo j,j}^{2,2}A_{5}^{2}A_{4}^{2}A_{2}^{2}A_{1}^{2}A_{i}^{2}A_{3}^{2}A_{j}^{2}\\
            &\:\:-A_{i}^{2}A_{6}^{2}F_{j\totwo i,j}^{2,2}A_{2}^{2}A_{1}^{2}A_{i}^{2}A_{3}^{2}A_{j}^{2}
            -F_{i\totwo j,i}A_{5}^{2}A_{4}^{2}A_{j}^{2}A_{2}^{2}A_{1}^{2}A_{i}^{2}A_{3}^{2}A_{j}^{2}\\
            &=A_{i}^{2}A_{j}^{2}\underline{A_{6}^{2}A_{i}^{2}}A_{5}^{2}A_{4}^{2}A_{2}^{2}A_{1}^{2}A_{3}^{2}A_{j}^{2}
            -A_{i}^{2}A_{j}^{2}A_{6}^{2}F_{j\totwo i,i}^{2,2}A_{2}^{2}A_{1}^{2}A_{3}^{2}A_{j}^{2}
            -A_{i}^{2}A_{j}^{2}A_{6}^{2}A_{5}^{2}A_{4}^{2}F_{i\toone j,i}^{2,2}A_{3}^{2}A_{j}^{2}\\
            &\:\:-A_{i}^{2}F_{i\totwo j,j}^{2,2}A_{5}^{2}A_{4}^{2}A_{2}^{2}A_{1}^{2}A_{i}^{2}A_{3}^{2}A_{j}^{2}
            -A_{i}^{2}A_{6}^{2}F_{j\totwo i,j}^{2,2}A_{2}^{2}A_{1}^{2}A_{i}^{2}A_{3}^{2}A_{j}^{2}
            -F_{i\totwo j,i}A_{5}^{2}A_{4}^{2}A_{j}^{2}A_{2}^{2}A_{1}^{2}A_{i}^{2}A_{3}^{2}A_{j}^{2}\\
            &=A_{i}^{2}A_{j}^{2}A_{i}^{2}A_{6}^{2}A_{5}^{2}A_{4}^{2}A_{2}^{2}A_{1}^{2}\underline{A_{3}^{2}A_{j}^{2}}
            -A_{i}^{2}A_{j}^{2}F_{i\totwo j,i}^{2,2}A_{5}^{2}A_{4}^{2}A_{2}^{2}A_{1}^{2}A_{3}^{2}A_{j}^{2}
            -A_{i}^{2}A_{j}^{2}A_{6}^{2}F_{j\totwo i,i}^{2,2}A_{2}^{2}A_{1}^{2}A_{3}^{2}A_{j}^{2}\\
            &\:\:-A_{i}^{2}A_{j}^{2}A_{6}^{2}A_{5}^{2}A_{4}^{2}F_{i\toone j,i}^{2,2}A_{3}^{2}A_{j}^{2}
            -A_{i}^{2}F_{i\totwo j,j}^{2,2}A_{5}^{2}A_{4}^{2}A_{2}^{2}A_{1}^{2}A_{i}^{2}A_{3}^{2}A_{j}^{2}
            -A_{i}^{2}A_{6}^{2}F_{j\totwo i,j}^{2,2}A_{2}^{2}A_{1}^{2}A_{i}^{2}A_{3}^{2}A_{j}^{2}\\
            &\:\:-F_{i\totwo j,i}A_{5}^{2}A_{4}^{2}A_{j}^{2}A_{2}^{2}A_{1}^{2}A_{i}^{2}A_{3}^{2}A_{j}^{2}\\
            &=A_{i}^{2}A_{j}^{2}A_{i}^{2}A_{6}^{2}A_{5}^{2}A_{4}^{2}\underline{A_{2}^{2}A_{1}^{2}A_{j}^{2}}A_{3}^{2}
            -A_{i}^{2}A_{j}^{2}A_{i}^{2}A_{6}^{2}A_{5}^{2}A_{4}^{2}A_{2}^{2}A_{1}^{2}F_{j\toone i,j}^{2,2}
            -A_{i}^{2}A_{j}^{2}F_{i\totwo j,i}^{2,2}A_{5}^{2}A_{4}^{2}A_{2}^{2}A_{1}^{2}A_{3}^{2}A_{j}^{2}\\
            &\:\:-A_{i}^{2}A_{j}^{2}A_{6}^{2}F_{j\totwo i,i}^{2,2}A_{2}^{2}A_{1}^{2}A_{3}^{2}A_{j}^{2}
            -A_{i}^{2}A_{j}^{2}A_{6}^{2}A_{5}^{2}A_{4}^{2}F_{i\toone j,i}^{2,2}A_{3}^{2}A_{j}^{2}
            -A_{i}^{2}F_{i\totwo j,j}^{2,2}A_{5}^{2}A_{4}^{2}A_{2}^{2}A_{1}^{2}A_{i}^{2}A_{3}^{2}A_{j}^{2}\\
            &\:\:-A_{i}^{2}A_{6}^{2}F_{j\totwo i,j}^{2,2}A_{2}^{2}A_{1}^{2}A_{i}^{2}A_{3}^{2}A_{j}^{2}
            -F_{i\totwo j,i}A_{5}^{2}A_{4}^{2}A_{j}^{2}A_{2}^{2}A_{1}^{2}A_{i}^{2}A_{3}^{2}A_{j}^{2}\\
            &=A_{i}^{2}A_{j}^{2}A_{i}^{2}A_{6}^{2}\underline{A_{5}^{2}A_{4}^{2}A_{j}^{2}}A_{2}^{2}A_{1}^{2}A_{3}^{2}
            -A_{i}^{2}A_{j}^{2}A_{i}^{2}A_{6}^{2}A_{5}^{2}A_{4}^{2}F_{i\toone j,j}^{2,2}A_{3}^{2}
            -A_{i}^{2}A_{j}^{2}F_{i\totwo j,i}^{2,2}A_{5}^{2}A_{4}^{2}A_{2}^{2}A_{1}^{2}A_{3}^{2}A_{j}^{2}\\
            &\:\:-A_{i}^{2}A_{j}^{2}A_{i}^{2}A_{6}^{2}A_{5}^{2}A_{4}^{2}A_{2}^{2}A_{1}^{2}F_{j\toone i,j}^{2,2}
            -A_{i}^{2}A_{j}^{2}A_{6}^{2}F_{j\totwo i,i}^{2,2}A_{2}^{2}A_{1}^{2}A_{3}^{2}A_{j}^{2}
            -A_{i}^{2}A_{j}^{2}A_{6}^{2}A_{5}^{2}A_{4}^{2}F_{i\toone j,i}^{2,2}A_{3}^{2}A_{j}^{2}\\
            &\:\:-A_{i}^{2}F_{i\totwo j,j}^{2,2}A_{5}^{2}A_{4}^{2}A_{2}^{2}A_{1}^{2}A_{i}^{2}A_{3}^{2}A_{j}^{2}
            -A_{i}^{2}A_{6}^{2}F_{j\totwo i,j}^{2,2}A_{2}^{2}A_{1}^{2}A_{i}^{2}A_{3}^{2}A_{j}^{2}
            -F_{i\totwo j,i}A_{5}^{2}A_{4}^{2}A_{j}^{2}A_{2}^{2}A_{1}^{2}A_{i}^{2}A_{3}^{2}A_{j}^{2}\\
            &=A_{i}^{2}A_{j}^{2}A_{i}^{2}\underline{A_{6}^{2}A_{j}^{2}}A_{5}^{2}A_{4}^{2}A_{2}^{2}A_{1}^{2}A_{3}^{2}
            -A_{i}^{2}A_{j}^{2}A_{i}^{2}A_{6}^{2}F_{j\totwo i,j}^{2,2}A_{2}^{2}A_{1}^{2}A_{3}^{2}
            -A_{i}^{2}A_{j}^{2}F_{i\totwo j,i}^{2,2}A_{5}^{2}A_{4}^{2}A_{2}^{2}A_{1}^{2}A_{3}^{2}A_{j}^{2}\\
            &\:\:-A_{i}^{2}A_{j}^{2}A_{i}^{2}A_{6}^{2}A_{5}^{2}A_{4}^{2}F_{i\toone j,j}^{2,2}A_{3}^{2}
            -A_{i}^{2}A_{j}^{2}A_{i}^{2}A_{6}^{2}A_{5}^{2}A_{4}^{2}A_{2}^{2}A_{1}^{2}F_{j\toone i,j}^{2,2}
            -A_{i}^{2}A_{j}^{2}A_{6}^{2}F_{j\totwo i,i}^{2,2}A_{2}^{2}A_{1}^{2}A_{3}^{2}A_{j}^{2}\\
            &\:\:-A_{i}^{2}A_{j}^{2}A_{6}^{2}A_{5}^{2}A_{4}^{2}F_{i\toone j,i}^{2,2}A_{3}^{2}A_{j}^{2}
            -A_{i}^{2}F_{i\totwo j,j}^{2,2}A_{5}^{2}A_{4}^{2}A_{2}^{2}A_{1}^{2}A_{i}^{2}A_{3}^{2}A_{j}^{2}
            -A_{i}^{2}A_{6}^{2}F_{j\totwo i,j}^{2,2}A_{2}^{2}A_{1}^{2}A_{i}^{2}A_{3}^{2}A_{j}^{2}\\
            &\:\:-F_{i\totwo j,i}A_{5}^{2}A_{4}^{2}A_{j}^{2}A_{2}^{2}A_{1}^{2}A_{i}^{2}A_{3}^{2}A_{j}^{2}\\
            &=\underbrace{A_{i}^{2}A_{j}^{2}A_{i}^{2}A_{j}^{2}}_{\A^{m}}A_{6}^{2}A_{5}^{2}A_{4}^{2}A_{2}^{2}A_{1}^{2}A_{3}^{2}
            -A_{i}^{2}A_{j}^{2}A_{i}^{2}F_{i\totwo j,j}^{2,2}A_{5}^{2}A_{4}^{2}A_{2}^{2}A_{1}^{2}A_{3}^{2}
            -A_{i}^{2}A_{j}^{2}A_{i}^{2}A_{6}^{2}F_{j\totwo i,j}^{2,2}A_{2}^{2}A_{1}^{2}A_{3}^{2}\\
            &\:\:-A_{i}^{2}A_{j}^{2}F_{i\totwo j,i}^{2,2}A_{5}^{2}A_{4}^{2}A_{2}^{2}A_{1}^{2}A_{3}^{2}A_{j}^{2}
            -A_{i}^{2}A_{j}^{2}A_{i}^{2}A_{6}^{2}A_{5}^{2}A_{4}^{2}F_{i\toone j,j}^{2,2}A_{3}^{2}
            -A_{i}^{2}A_{j}^{2}A_{i}^{2}A_{6}^{2}A_{5}^{2}A_{4}^{2}A_{2}^{2}A_{1}^{2}F_{j\toone i,j}^{2,2}\\
            &\:\:-A_{i}^{2}A_{j}^{2}A_{6}^{2}F_{j\totwo i,i}^{2,2}A_{2}^{2}A_{1}^{2}A_{3}^{2}A_{j}^{2}
            -A_{i}^{2}A_{j}^{2}A_{6}^{2}A_{5}^{2}A_{4}^{2}F_{i\toone j,i}^{2,2}A_{3}^{2}A_{j}^{2}
            -A_{i}^{2}F_{i\totwo j,j}^{2,2}A_{5}^{2}A_{4}^{2}A_{2}^{2}A_{1}^{2}A_{i}^{2}A_{3}^{2}A_{j}^{2}\\
            &\:\:-A_{i}^{2}A_{6}^{2}F_{j\totwo i,j}^{2,2}A_{2}^{2}A_{1}^{2}A_{i}^{2}A_{3}^{2}A_{j}^{2}
            -F_{i\totwo j,i}A_{5}^{2}A_{4}^{2}A_{j}^{2}A_{2}^{2}A_{1}^{2}A_{i}^{2}A_{3}^{2}A_{j}^{2}.)
        \end{align*}

        Now, from the sub-multiplicativity and sub-additivity properties of the induced norm, we have
        \begin{align}
        \label{e:pf1_step1}
            \norm{\overline{W}} &= \norm{(\A^{m}L_{1}+L_{2})R}
            \leq \norm{\A^{m}}\norm{L_{1}R} + \norm{L_{2}}\norm{R}\nonumber\\
            &\leq \rho c e^{-\lambda(\abs{\overline{W}}-m)} + \Biggl(\frac{\mb(\mb-1)}{2}M^{\xi_{j\overset{1}{\to} i,i}}\varepsilon_{j\overset{1}{\to} i,i}
            +\frac{m(m-1)-\mb(\mb-1)}{2}M^{\xi_{j\overset{2}{\to} i,i}}\varepsilon_{j\overset{2}{\to} i,i}\nonumber\\
            &\quad\quad+\frac{\mb(\mb+1)}{2}M^{\xi_{j\overset{1}{\to} i,j}}\varepsilon_{j\overset{1}{\to} i,j}
            +\frac{m(m+1)-\mb(\mb+1)}{2}M^{\xi_{j\overset{2}{\to} i,j}}\varepsilon_{j\overset{2}{\to} i,j}\nonumber\\
            &\quad\quad+\frac{\mb(\mb+1)}{2}M^{\xi_{i\overset{1}{\to} j,i}}\varepsilon_{i\overset{1}{\to} j,i}
            +\frac{m(m+1)-\mb(\mb+1)}{2}M^{\xi_{i\overset{2}{\to} j,i}}\varepsilon_{i\overset{2}{\to} j,i}\nonumber\\
            &\quad\quad+\frac{\mb(\mb+1)}{2}M^{\xi_{i\overset{1}{\to} j,j}}\varepsilon_{i\overset{1}{\to} j,j}
            +\frac{m(m+1)-\mb(\mb+1)}{2}M^{\xi_{i\overset{2}{\to} j,j}}\varepsilon_{i\overset{2}{\to} j,j}\Biggr)\nonumber\\
            &\qquad\qquad\times ce^{-\lambda\biggl(\abs{\overline{W}}-\biggl(\xi_{j\overset{1}{\to}i\overset{1}{\to}j}
            +\xi_{j\overset{2}{\to}i\overset{2}{\to}j}\biggr)\biggr)}\nonumber\\
            &=ce^{-\lambda\abs{\overline{W}}}\Biggl(\rho e^{\lambda m} + \Biggl(\frac{\mb(\mb-1)}{2}M^{\xi_{j\overset{1}{\to} i,i}}\varepsilon_{j\overset{1}{\to} i,i}
            +\frac{m(m-1)-\mb(\mb-1)}{2}M^{\xi_{j\overset{2}{\to} i,i}}\varepsilon_{j\overset{2}{\to} i,i}\nonumber\\
            &\quad\quad+\frac{\mb(\mb+1)}{2}M^{\xi_{j\overset{1}{\to} i,j}}\varepsilon_{j\overset{1}{\to} i,j}
            +\frac{m(m+1)-\mb(\mb+1)}{2}M^{\xi_{j\overset{2}{\to} i,j}}\varepsilon_{j\overset{2}{\to} i,j}\nonumber\\
            &\quad\quad+\frac{\mb(\mb+1)}{2}M^{\xi_{i\overset{1}{\to} j,i}}\varepsilon_{i\overset{1}{\to} j,i}
            +\frac{m(m+1)-\mb(\mb+1)}{2}M^{\xi_{i\overset{2}{\to} j,i}}\varepsilon_{i\overset{2}{\to} j,i}\nonumber\\
            &\quad\quad+\frac{\mb(\mb+1)}{2}M^{\xi_{i\overset{1}{\to} j,j}}\varepsilon_{i\overset{1}{\to} j,j}
            +\frac{m(m+1)-\mb(\mb+1)}{2}M^{\xi_{i\overset{2}{\to} j,j}}\varepsilon_{i\overset{2}{\to} j,j}\Biggr)\nonumber\\
            &\qquad\qquad\times e^{\lambda\biggl(\xi_{j\overset{1}{\to}i\overset{1}{\to}j}+\xi_{j\overset{2}{\to}i\overset{2}{\to}j}\biggr)}\Biggr).
        \end{align}
        In the above inequality the upper bounds on \(\norm{L_{1}R}\) and \(\norm{R}\) are obtained from the relations \(\abs{\overline{W}} = \abs{\A^{m}}+\abs{L_{1}R}\) and \(\abs{\overline{W}} = \abs{L}+\abs{R}\), respectively. Applying \eqref{e:maincondn4} to \eqref{e:pf1_step1} leads to \eqref{e:ges2}. Consequently, \eqref{e:swsys} is GES under \(\sigma\).

        This completes our proof of Theorem \ref{t:mainres1}.
   \end{proof}

        Given the admissible switches between the subsystems, \(E(\P)\), and the admissible minimum and maximum dwell times on the subsystems, \(\delta\) and \(\Delta\), Theorem \ref{t:mainres1} provides sufficient conditions on the subsystems matrices, \(A_{\ell}\), \(\ell\in\P\), such that there exists a switching signal \(\sigma\in\Sw\) under which the switched system \eqref{e:swsys} is GES. Since \(\rho < 1\), it is always possible to find a \(\lambda > 0\) (could be small) such that condition \eqref{e:maincondn1} holds. If, in addition, the underlying directed graph, \(G(\P,E(\P))\), of the switched system \eqref{e:swsys}, admits two pairs of paths, \((P_{G}^{j\tor i},P_{G}^{i\tor j})\), \(r=1,2\), \(i,j\in\P\) satisfying Assumption \ref{a:stable_combi}, \(P_{G}^{j\overset{r}{\to} i} = w_{0}^{(r)},(w_{0}^{(r)},w_{1}^{(r)}),w_{1}^{(r)}\),\\\(\ldots,w_{\ell_{r}-1}^{(r)},(w_{\ell_{r}-1}^{(r)},w_{\ell_{r}})\),
   \(w_{\ell_{r}}\) and \(P_{G}^{i\overset{r}{\to} j} = w_{\ell_{r}}^{(r)},(w_{\ell_{r}}^{(r)},w_{\ell_{r}+1}^{(r)}),w_{\ell_{r}+1}^{(r)},\ldots,w_{n_{r}-1}^{(r)},(w_{n_{r}-1}^{(r)},
   w_{n_{r}})\),\\\( w_{n_{r}}^{(r)}\), \(r=1,2\), for which the Euclidean norms of the commutators \(F_{j\overset{r}{\to}i,i}^{p,\delta}\), \(F_{j\overset{r}{\to}i,j}^{q,\delta}\), \(F_{i\overset{r}{\to}j,i}^{p,\delta}\) and \(F_{i\overset{r}{\to}j,j}^{q,\delta}\) are bounded above by small enough scalars \(\varepsilon_{j\overset{r}{\to}i,i}\), \(\varepsilon_{j\overset{r}{\to}i,j}\), \(\varepsilon_{i\overset{r}{\to}j,i}\) and \(\varepsilon_{i\overset{r}{\to}j,j}\), respectively, \(r=1,2\), such that condition \eqref{e:maincondn4} holds, then there exists a \(\sigma\in\Sw\) under which \eqref{e:swsys} is GES. This switching signal activates the sequence of subsystems, \(w_{0}^{(1)},w_{1}^{(1)},\ldots,w_{\ell_{1}}^{(1)},w_{\ell_{1}+1}^{(1)},\ldots,w_{n_{1}-1}^{(1)}\), followed by \(s\)-many instances of the sequence of subsystems, \(w_{0}^{(2)},w_{1}^{(2)},\ldots,w_{\ell_{2}}^{(2)},w_{\ell_{2}+1}^{(2)},\ldots,w_{n_{2}-1}^{(2)}\), repeatedly, \(s=1,2,3,\ldots\), and dwells on the subsystems \(i\), \(j\) and \(k\in\P\setminus\{i,j\}\) for \(p\), \(q\) and \(\delta\) units of time, respectively. Clearly, \(\sigma\) is a non-periodic element of \(\Sw\). Notice that Theorem \ref{t:mainres1} does not require \((i,j)\) and/or \((j,i)\in E(\P)\) for the utilization of stability of the matrix product, \(A_{i}^{p}A_{j}^{q}\), and works as long as it is possible to reach from subsystem \(j\) to subsystem \(i\) and vice-versa, through multiple paths that satisfy conditions \eqref{e:maincondn2}-\eqref{e:maincondn4}.

   \begin{example}
   \label{ex:numex1}
   \rm{
         Let
         \begin{align*}
            P_{G}^{j\toone i} &= j,(j,3),3,(3,i),i,\\
            P_{G}^{j\totwo i} &= j,(j,4),4,(4,5),5,(5,i),i,\\
            P_{G}^{i\toone j} &= i,(i,1),1,(1,2),2,(2,j),j,\\
            \intertext{and}
            P_{G}^{i\totwo j} &= i,(i,6),6,(6,j),j
         \end{align*}
         satisfy conditions \eqref{e:maincondn2}-\eqref{e:maincondn4}. Then a switching signal that activates the sequences of subsystems
        \begin{align*}
            &{j,3,i,1,2,j,4,5,i,6},\\
            &{j,3,i,1,2,j,4,5,i,6,j,4,5,i,6},\\
            &{j,3,i,1,2,j,4,5,i,6,j,4,5,i,6,j,4,5,i,6},\\
            &\vdots,
        \end{align*}
        and dwells on the subsystems \(i\), \(j\) and \(k\in\{1,2,3,4,5,6\}\) for \(p\), \(q\) and \(\delta\) units of time, respectively, ensures GES of the switched system \eqref{e:swsys}.
        }
   \end{example}

    \begin{rem}
    \label{rem:robustness}
    \rm{
        Condition \eqref{e:maincondn4} involves the rate of decay of the Schur stable matrix, \(A_{i}^{p}A_{j}^{q}\), upper bounds on the Euclidean norms of the commutators of certain products of the matrices, \(A_{k}^{\delta}\), \(k\in\P\setminus\{i,j\}\), that appear in the paths \(P_{G}^{j\overset{r}{\to} i}\) (resp., \(P_{G}^{i\overset{r}{\to} j}\)), \(r=1,2\) and the matrix product \(A_{j}^{q}\) (resp., \(A_{i}^{p}\)), and a set of scalars capturing the properties of the matrices, \(A_{\ell}\), \(\ell\in\P\). In the simplest case when the matrix products, \(A_{w_{\ell_{r}-1}^{(r)}}^{\delta}\ldots A_{w_{1}^{(r)}}^{\delta}\) and \(A_{j}^{q}\), and the matrix products \(A_{w_{n_{r}-1}^{(r)}}^{\delta}\ldots A_{w_{\ell_{r}+1}^{(r)}}^{\delta}\) and \(A_{i}^{p}\) commute, \(r=1,2\), the condition \eqref{e:maincondn4} reduces to \eqref{e:maincondn1}. Theorem \ref{t:mainres1}, therefore, accommodates subsystems matrices, \(A_{\ell}\), \(\ell\in\P\), for which the above matrix products do not necessarily commute, but are close to matrices, \(\overline{A}_{\ell},\:\:\ell\in\P\), for which they commute. This feature associates an inherent robustness with our stabilizability conditions. Indeed, if we are relying on approximate models of \(A_{\ell}\), \(\ell\in\P\), or the elements of \(A_{\ell}\), \(\ell\in\P\), are prone to evolve over time, then GES of \eqref{e:swsys} is preserved under our set of switching signals as long as Assumption \ref{a:stable_combi} and conditions \eqref{e:maincondn2}-\eqref{e:maincondn4} continue to hold.
        }
    \end{rem}

    A next natural topic of discussion is regarding the setting where \(G(\P,E(\P))\) does not admit two pairs of paths between the subsystems that form a Schur stable combination and/or it admits pairs of paths that are not favourable in the sense of Theorem \ref{t:mainres1}. We show, as a special case of Theorem \ref{t:mainres1}, that it is possible to stabilize \eqref{e:swsys} even if \(G(\P,E(\P))\) admits exactly one pair containing a \(j\to i\) path and an \(i\to j\) path, that satisfies a set of scalar inequalities similar to \eqref{e:maincondn2}-\eqref{e:maincondn4}. In particular, we will ensure GES of \eqref{e:swsys} under a periodic element of \(\Sw\) in this setting.

    For \(P_{G}^{j\toone i} = P_{G}^{j\totwo i} = P_{G}^{j\to i}\) and \(P_{G}^{i\toone j} = P_{G}^{i\totwo j} = P_{G}^{i\to j}\), we define
   \begin{align*}
        \zeta_{j\to i,i} &= \abs{P_{G}^{j\to i}}\delta(m-1)+\abs{P_{G}^{i\to j}}\delta m+p(m-1)+qm,\nonumber\\
        \zeta_{j\to i,j} &= \abs{P_{G}^{j\to i}}\delta(m-1) + \abs{P_{G}^{i\to j}}\delta m + pm + q(m-1),\nonumber\\
        \zeta_{i\to j,i} &= \abs{P_{G}^{j\to i}}\delta m + \abs{P_{G}^{i\to j}}\delta(m-1) + p(m-1) + qm,\\
        \zeta_{i\to j,j} &= \abs{P_{G}^{j\to i}}\delta m + \abs{P_{G}^{i\to j}}\delta(m-1) + pm + q(m-1),\nonumber\\
        \zeta_{j\to i\to j} &= \Biggl(\biggl(\abs{P_{G}^{j\to i}}+\abs{P_{G}^{i\to j}}\biggr)\delta+p+q\Biggr)m,\nonumber
   \end{align*}

   \begin{corollary}
   \label{cor:mainres2}
        Let \(i,j\in\P\) satisfy Assumption \ref{a:stable_combi} and \(\lambda\) be an arbitrary positive number satisfying \eqref{e:maincondn1}. Suppose that \(G(\P,E(\P))\) admits a pair of \(j\to i\) and \(i\to j\) paths, \((P_{G}^{j\to i},P_{G}^{i\to j})\), that satisfy the following condition: there exist scalars \(\varepsilon_{j\to i,i}\), \(\varepsilon_{j\to i,j}\), \(\varepsilon_{i\to j,i}\) and \(\varepsilon_{i\to j,j}\), small enough, such that
        \begin{align}
        \label{e:maincondn5}
            \norm{F_{j\to i,i}^{p,\delta}}\leq\varepsilon_{j\to i,i}\:\:\text{and}\:\:\norm{F_{j\to i,j}^{q,\delta}}\leq \varepsilon_{j\to i,j},
        \end{align}
        \begin{align}
        \label{e:maincondn6}
            \norm{F_{i\to j,i}^{p,\delta}}\leq \varepsilon_{i\to j,i}\:\:\text{and}\:\:\norm{F_{i\to j,j}^{q,\delta}}\leq \varepsilon_{i\to j,j},
        \end{align}
        and
        \begin{align}
        \label{e:maincondn7}
            \rho e^{\lambda m} + \Biggl(&\frac{m(m-1)}{2}M^{\zeta_{j\to i,i}}\varepsilon_{j\to i,i}+\frac{m(m+1)}{2}\biggl(M^{\zeta_{j\to i,j}}\varepsilon_{j\to i,j}
            +M^{\zeta_{i\to j,i}}\varepsilon_{i\to j,i}\nonumber\\
            &\qquad\qquad+M^{\zeta_{i\to j,j}}\varepsilon_{i\to j,j}\biggr)\Biggr)\times e^{\lambda\zeta_{j\to i\to j}}\leq 1.
        \end{align}
        Then there exists a periodic switching signal \(\sigma\in\Sw\) under which the switched system \eqref{e:swsys} is GES.
   \end{corollary}

   \begin{proof}
        We apply Theorem \ref{t:mainres1} to arrive at the assertion of Corollary \ref{cor:mainres2}.

        Let
        \begin{align}
        \label{e:pf2_step1}
            P_{G}^{j\toone i} &= P_{G}^{j\totwo i} = P_{G}^{j\to i} = w_{0},(w_{0},w_{1}),w_{1},\ldots,w_{\ell-1},(w_{\ell-1},w_{\ell}),w_{\ell},\nonumber\\
            P_{G}^{i\toone j} &= P_{G}^{i\totwo j} = P_{G}^{i\to j} = w_{\ell},(w_{\ell},w_{\ell+1}),w_{\ell+1},\ldots,w_{n-1},
            (w_{n-1},w_{n}),w_{n}
        \end{align}
        satisfy conditions \eqref{e:maincondn5}-\eqref{e:maincondn7}. Consider the switching signal \(\sigma\in\Sw\) constructed in the proof of Theorem \ref{t:mainres1}. In view of \eqref{e:pf2_step1}, \(\sigma\) is periodic. Indeed, \(\sigma\) activates the sequence of subsystems \(w_{0},w_{1},\ldots,w_{\ell},w_{\ell+1},\ldots,w_{n-1}\) repeatedly and dwells on the subsystems \(i\), \(j\) and \(k\in\P\setminus\{i,j\}\) for \(p\), \(q\) and \(\delta\) units of time, respectively. Mathematically speaking, \(\sigma\) satisfies
        \begin{align*}
            &\sigma(\tau_{0}) = w_{0},\sigma(\tau_{1}) = w_{1},\ldots,\sigma(\tau_{\ell})=w_{\ell},\sigma(\tau_{\ell+1})=w_{\ell+1},\ldots,\\
            &\sigma(\tau_{n-1})=w_{n-1},
            \sigma(\tau_{n}) = w_{0},\ldots, \sigma(\tau_{2n-1}) = w_{n-1},
            \sigma(\tau_{2n}) = w_{0},\\
            &\ldots
        \end{align*}
        with
        \begin{align*}
            \tau_{h+1}-\tau_{h} =
            \begin{cases}
                p,\:\:&\text{if}\:\sigma(\tau_{h}) = i,\\
                q,\:\:&\text{if}\:\sigma(\tau_{h}) = j,\\
                \delta,\:\:&\text{if}\:\sigma(\tau_{h})\in\P\setminus\{i,j\}.
            \end{cases}
        \end{align*}

        We now show that the switched system \eqref{e:swsys} is GES under the above \(\sigma\). We have
        \begin{align*}
            F_{j\toone i,i}^{p,\delta} &= F_{j\totwo i,i}^{p,\delta} = F_{j\to i,i}^{p,\delta},\:\:
            F_{j\toone i,j}^{q,\delta} = F_{j\totwo i,j}^{q,\delta} = F_{j\to i,j}^{q,\delta},\\
            F_{i\toone j,i}^{p,\delta} &= F_{i\totwo j,i}^{p,\delta} = F_{i\to j,i}^{p,\delta},\:\:
            F_{i\toone j,j}^{q,\delta} = F_{i\totwo j,j}^{q,\delta} = F_{i\to j,j}^{q,\delta},
        \end{align*}
        leading to
        \begin{align*}
            \varepsilon_{j\toone i,i} &= \varepsilon_{j\totwo i,i} = \varepsilon_{j\to i,i},\:\:
            \varepsilon_{j\toone i,j} = \varepsilon_{j\totwo i,j} = \varepsilon_{j\to i,j},\\
            \varepsilon_{i\toone j,i} &= \varepsilon_{i\totwo j,i} = \varepsilon_{i\to j,i},\:\:
            \varepsilon_{i\toone j,j} = \varepsilon_{i\totwo j,j} = \varepsilon_{i\to j,j}.
        \end{align*}
        Moreover,
        \begin{align*}
            &\xi_{j\toone i,i} = \abs{P_{G}^{j\to i}}\delta(m-1) + \abs{P_{G}^{i\to j}}\delta m + p(m-1) + qm = \zeta_{j\to i,i},\\
            &\xi_{j\totwo i,i} = \abs{P_{G}^{j\to i}}\delta(m-1) + \abs{P_{G}^{i\to j}}\delta m + p(m-1) + qm = \zeta_{j\to i,i},\\
            &\xi_{j\toone i,j} = \abs{P_{G}^{j\to i}}\delta(m-1) + \abs{P_{G}^{i\to j}}\delta m + pm + q(m-1) = \zeta_{j\to i,j},\\
            &\xi_{j\totwo i,j} = \abs{P_{G}^{j\to i}}\delta(m-1) + \abs{P_{G}^{i\to j}}\delta m + pm + q(m-1) = \zeta_{j\to i,j},\\
            &\xi_{i\toone j,i} = \abs{P_{G}^{j\to i}}\delta m + \abs{P_{G}^{i\to j}}\delta (m-1) + p(m-1) + qm = \zeta_{i\to j,i},\\
            &\xi_{i\totwo j,i} = \abs{P_{G}^{j\to i}}\delta m + \abs{P_{G}^{i\to j}}\delta (m-1) + p(m-1) + qm = \zeta_{i\to j,i},\\
            &\xi_{i\toone j,j} = \abs{P_{G}^{j\to i}}\delta m + \abs{P_{G}^{i\to j}}\delta (m-1) + p m + q(m-1) = \zeta_{i\to j,j},\\
            &\xi_{i\totwo j,j} = \abs{P_{G}^{j\to i}}\delta m + \abs{P_{G}^{i\to j}}\delta (m-1) + pm + q(m-1) = \zeta_{i\to j,j},\\
            &\xi_{j\toone i\toone j} = \abs{P_{G}^{j\to i}}\delta\mb + \abs{P_{G}^{i\to j}}\delta\mb + (p+q)\mb,\\
            &\xi_{j\totwo i\totwo j} = \abs{P_{G}^{j\to i}}\delta m - \abs{P_{G}^{j\to i}}\delta\mb + \abs{P_{G}^{i\to j}}\delta m - \abs{P_{G}^{i\to j}}\delta\mb
            + (p+q)m - (p+q)\mb,\\
            \intertext{and}
            &\xi_{j\toone i\toone j} + \xi_{j\totwo i\totwo j} = \Biggl(\biggl(\abs{P_{G}^{j\to i}}+\abs{P_{G}^{i\to j}}\biggr)\delta+p+q\Biggr)m = \zeta_{j\to i\to j}.
        \end{align*}

        Consequently, the left-hand side of \eqref{e:maincondn4} becomes
        \begin{align*}
            &\rho e^{\lambda m} + \Biggl(\frac{\mb(\mb-1)}{2}M^{\zeta_{j{\to} i,i}}\varepsilon_{j{\to} i,i}
            +\frac{m(m-1)-\mb(\mb-1)}{2}M^{\zeta_{j{\to} i,i}}\varepsilon_{j{\to} i,i}\\
            &\quad\quad\quad+\frac{\mb(\mb+1)}{2}M^{\zeta_{j{\to} i,j}}\varepsilon_{j{\to} i,j}
            +\frac{m(m+1)-\mb(\mb+1)}{2}M^{\zeta_{j{\to} i,j}}\varepsilon_{j{\to} i,j}\\
            &\quad\quad\quad+\frac{\mb(\mb+1)}{2}M^{\zeta_{i{\to} j,i}}\varepsilon_{i{\to} j,i}
            +\frac{m(m+1)-\mb(\mb+1)}{2}M^{\zeta_{i{\to} j,i}}\varepsilon_{i{\to} j,i}\\
            &\quad\quad\quad+\frac{\mb(\mb+1)}{2}M^{\zeta_{i{\to} j,j}}\varepsilon_{i{\to} j,j}
            +\frac{m(m+1)-\mb(\mb+1)}{2}M^{\zeta_{i{\to} j,j}}\varepsilon_{i{\to} j,j}\Biggr)
            \times e^{\lambda\zeta_{j\to i\to j}}\\
            =&\rho e^{\lambda m} + \Biggl(\frac{m(m-1)}{2}M^{\zeta_{j\to i,i}}\varepsilon_{j\to i,i}+\frac{m(m+1)}{2}M^{\zeta_{j\to i,j}}\varepsilon_{j\to i,j}\\
            &\quad\quad\quad+\frac{m(m+1)}{2}M^{\zeta_{i\to j,i}}\varepsilon_{i\to j,i}+\frac{m(m+1)}{2}M^{\zeta_{i\to j,j}}\varepsilon_{i\to j,j}\Biggr)\times e^{\lambda\zeta_{j\to i\to j}}.
        \end{align*}
        In view of condition \eqref{e:maincondn7}, the above expression is at most equal to \(1\). It follows by the assertion of Theorem \ref{t:mainres1} that the switched system \eqref{e:swsys} is stabilizable.
   \end{proof}

        Corollary \ref{cor:mainres2} asserts that if the underlying directed graph, \(G(\P,E(\P))\), of the switched system \eqref{e:swsys} admits a pair of paths, \((P_{G}^{j\to i},P_{G}^{i\to j})\), \(i,j\in\P\) satisfying Assumption \ref{a:stable_combi}, \(P_{G}^{j\to i} = w_{0},(w_{0},w_{1}),w_{1},\ldots,w_{\ell-1},(w_{\ell-1},w_{\ell}), w_{\ell}\) and \(P_{G}^{i\to j} = w_{\ell}, (w_{\ell},w_{\ell+1}),w_{\ell+1},\ldots,w_{n-1}\),\((w_{n-1},w_{n})\),\\\(w_{n}\), for which the Euclidean norms of (matrix) commutators \(F_{j\to i,i}^{p,\delta}\), \(F_{j\to i,j}^{q,\delta}\), \(F_{i\to j,i}^{p,\delta}\) and \(F_{i\to j,j}^{q,\delta}\) are bounded above by small enough scalars \(\varepsilon_{j\to i,i}\), \(\varepsilon_{j\to i,j}\), \(\varepsilon_{i\to j,i}\) and \(\varepsilon_{i\to j,j}\), respectively, such that condition \eqref{e:maincondn7} holds, then \eqref{e:swsys} is GES under a periodic switching signal \(\sigma\in\Sw\). The stabilizing switching signal activates the sequence of subsystems, \(w_{0},w_{1},\ldots,w_{\ell},\ldots,w_{n-1}\) repeatedly and dwells on the subsystems \(i\), \(j\) and \(k\in\P\setminus\{i,j\}\) for \(p\), \(q\) and \(\delta\) units of time, respectively.

    \begin{example}
    \label{rem:numex2}
    \rm{
         Let
         \begin{align*}
            P_{G}^{j\to i} &= j,(j,2),2,(2,1),1,(1,i),i\\
             \intertext{and}
            P_{G}^{i\to j} &= i,(i,1),1,(1,j),j
         \end{align*}
         satisfy conditions \eqref{e:maincondn5}-\eqref{e:maincondn7}. Then the switching signal that activates the sequence of subsystems
   \[
        j,2,1,i,1
   \]
   repeatedly and dwells on subsystems \(i\), \(j\) and \(k\in\{1,2\}\) for \(p\), \(q\) and \(\delta\) units of time, respectively,
   ensures GES of \eqref{e:swsys}.
   }
   \end{example}

   Similar to the setting of Theorem \ref{t:mainres1}, our stability conditions in Corollary \ref{cor:mainres2} do not assume \((i,j)\) and/or \((j,i)\in E(\P)\), and works as long as it is possible to reach subsystem \(i\) from subsystem \(j\) and vice-versa, through a pair of favourable paths. Of course, if the choice of \(i,j\in\P\) that satisfy Assumption \ref{a:stable_combi} is such that the switches between them are partially restricted or unrestricted, then one requires to check simpler conditions on the subsystems matrices, \(A_{\ell}\), \(\ell\in\P\), and the directed graph, \(G(\P,E(\P))\), for the stabilization of \eqref{e:swsys}. Below we enumerate these settings as special cases of Theorem \ref{t:mainres1} and Corollary \ref{cor:mainres2}.

   For two pairs of \(j\to i\) and \(i\to j\) paths, \((P_{G}^{j\toone i},P_{G}^{i\toone j})\) and \((P_{G}^{j\totwo i},P_{G}^{i\totwo j})\) on \(G(\P,E(\P))\), we define
    \begin{align*}
        \kappa_{j\toone i,i} &= \abs{P_{G}^{j\toone i}}\delta(\mb-1) + \abs{P_{G}^{j\totwo i}}\delta(m-\mb) + p(m-1) + qm,\nonumber\\
        \kappa_{j\totwo i,i} &= \abs{P_{G}^{j\toone i}}\delta\mb + \abs{P_{G}^{j\totwo i}}\delta(m-\mb-1) + p(m-1) + qm,\nonumber\\
        \kappa_{j\toone i,j} &= \abs{P_{G}^{j\toone i}}\delta(\mb -1) + \abs{P_{G}^{j\totwo i}}\delta(m-\mb) + pm + q(m-1),\\
        \kappa_{j\totwo i,j} &= \abs{P_{G}^{j\toone i}}\delta\mb + \abs{P_{G}^{j\totwo i}}\delta(m-\mb-1) + pm + q(m-1),\nonumber\\
        \kappa_{j\toone i\toone j} &= \biggl(\abs{P_{G}^{j\toone i}}\delta+p+q\biggr)\mb,\nonumber\\
        \kappa_{j\totwo i\totwo j} &= \biggl(\abs{P_{G}^{j\totwo i}}\delta+p+q\biggr)(m-\mb).\nonumber
   \end{align*}
   Also, for \(P_{G}^{j\toone i} = P_{G}^{j\totwo i} = P_{G}^{j\to i}\) and \(P_{G}^{i\toone j} = P_{G}^{i\totwo j} = P_{G}^{i\to j}\), we define the quantities
   \begin{align*}
        \kb_{j\to i,i} &= \abs{P_{G}^{j\to i}}\delta(m-1)+p(m-1)+qm,\nonumber\\
        \kb_{j\to i,j} &= \abs{P_{G}^{j\to i}}\delta(m-1)+pm+q(m-1),\\
        \kb_{j\to i\to j} &= \biggl(\abs{P_{G}^{j\to i}}\delta+p+q\biggr)m.\nonumber
   \end{align*}

   \begin{corollary}
   \label{cor:mainres3}
        Let \(i,j\in\P\) satisfy Assumption \ref{a:stable_combi}, \((i,j)\in E(\P)\), \((j,i)\notin E(\P)\), and \(\lambda\) be an arbitrary positive number satisfying \eqref{e:maincondn1}. The following hold:\\
             (a) Suppose that \(G(\P,E(\P))\) admits two \(j\to i\) paths, \(P_{G}^{j\tor i}\), \(r=1,2\), such that there exist scalars \(\varepsilon_{j\tor i,i}\), \(\varepsilon_{j\tor i,j}\), \(r=1,2\) , small enough, satisfying
                \begin{align}
                \label{e:maincondn8}
                    \norm{F_{j\tor i,i}^{p,\delta}}\leq\varepsilon_{j\tor i,i},\:\:
                    \norm{F_{j\tor i,j}^{q,\delta}}\leq\varepsilon_{j\tor i,j},\:\:r=1,2,
                \end{align}
                and
                \begin{align}
                \label{e:maincondn9}
                    &\rho e^{\lambda m} + \Biggl(\frac{\mb(\mb-1)}{2}M^{\kappa_{j\overset{1}{\to} i,i}}\varepsilon_{j\overset{1}{\to} i,i}
            +\frac{m(m-1)-\mb(\mb-1)}{2}M^{\kappa_{j\overset{2}{\to} i,i}}\varepsilon_{j\overset{2}{\to} i,i}\nonumber\\
            &\quad\quad\quad+\frac{\mb(\mb+1)}{2}M^{\kappa_{j\overset{1}{\to} i,j}}\varepsilon_{j\overset{1}{\to} i,j}
            +\frac{m(m+1)-\mb(\mb+1)}{2}M^{\kappa_{j\overset{2}{\to} i,j}}\varepsilon_{j\overset{2}{\to} i,j}\Biggr)\nonumber\\
            &\qquad\qquad\times e^{\lambda\biggl(\kappa_{j\overset{1}{\to}i\overset{1}{\to}j}+\kappa_{j\overset{2}{\to}i\overset{2}{\to}j}\biggr)}\leq 1.
                \end{align}
                Then there exists a non-periodic switching signal \(\sigma\in\Sw\) under which the switched system \eqref{e:swsys} is GES.\\
             (b) Suppose that \(G(\P,E(\P))\) admits a \(j\to i\) path, \(P_{G}^{j\to i}\), such that there exist scalars \(\varepsilon_{j\to i,i}\), \(\varepsilon_{j\to i,j}\), small enough, satisfying
                \begin{align}
                \label{e:maincondn10}
                    \norm{F_{j\to i,i}^{p,\delta}}\leq\varepsilon_{j\to i,i}\:\:\text{and}\:\:\norm{F_{j\to i,j}^{q,\delta}}\leq\varepsilon_{j\to i,j},
                \end{align}
                and
                \begin{align}
                \label{e:maincondn11}
                     \rho e^{\lambda m} + \Biggl(&\frac{m(m-1)}{2}M^{\kb_{j\to i,i}}\varepsilon_{j\to i,i}+\frac{m(m+1)}{2}M^{\kb_{j\to i,j}}\varepsilon_{j\to i,j}\Biggr)
                     \times e^{\lambda\kb_{j\to i\to j}}\leq 1.
                \end{align}
                Then there exists a periodic switching signal \(\sigma\in\Sw\) under which the switched system \eqref{e:swsys} is GES.
   \end{corollary}

   \begin{proof}
        Since \((i,j)\in E(\P)\), the graph \(G(\P,E(\P))\) admits an \(i\to j\) path, \(P_{G}^{i\to j} = i,(i,j),j\), that satisfies \(\abs{P_{G}^{i\to j}} = 0\), \(F_{i\to j,i} = 0\) and \(F_{i\to j,j} = 0\).

         Suppose that (a) holds. We have two \(j\to i\) paths, \(P_{G}^{j\toone i}\) and \(P_{G}^{j\totwo i}\), that satisfy conditions \eqref{e:maincondn8}-\eqref{e:maincondn9}. Let \(P_{G}^{i\toone j} = P_{G}^{i\totwo j} = P_{G}^{i\to j}\). Thus, the choice \(\varepsilon_{i\toone j,i} = \varepsilon_{i\totwo j,i} = \varepsilon_{i\toone j,j} = \varepsilon_{i\totwo j,j} = 0\) holds. We apply Theorem \ref{t:mainres1} to show that \eqref{e:swsys} is stabilizable in this setting.

        We have
        \begin{align*}
            \xi_{j\toone i,i} &= \abs{P_{G}^{j\toone i}}\delta(\mb-1)+\abs{P_{G}^{j\totwo i}}\delta(m-\mb)+p(m-1)+qm=\kappa_{j\toone i,i},\\
            \xi_{j\totwo i,i} &= \abs{P_{G}^{j\toone i}}\delta\mb+\abs{P_{G}^{j\totwo i}}\delta(m-\mb-1)+p(m-1)+qm=\kappa_{j\totwo i,i},\\
            \xi_{j\toone i,j} &= \abs{P_{G}^{j\toone i}}\delta(\mb-1)+\abs{P_{G}^{j\totwo i}}\delta(m-\mb)+pm+q(m-1)=\kappa_{j\toone i,j},\\
            \xi_{j\totwo i,j} &= \abs{P_{G}^{j\toone i}}\delta\mb+\abs{P_{G}^{j\totwo i}}\delta(m-\mb-1)+pm+q(m-1)=\kappa_{j\totwo i,j},\\
            \xi_{j\toone i\toone j} &= \biggl(\abs{P_{G}^{j\toone i}}\delta+p+q\biggr)\mb = \kappa_{j\toone i\toone j},\\
            \intertext{and}
            \xi_{j\totwo i\totwo j} &= \biggl(\abs{P_{G}^{j\totwo i}}\delta+p+q\biggr)(m-\mb) = \kappa_{j\totwo i\totwo j}.
        \end{align*}

        Consequently, the left-hand side of \eqref{e:maincondn4} becomes
        \begin{align*}
            &\rho e^{\lambda m} + \Biggl(\frac{\mb(\mb-1)}{2}M^{\kappa_{j\overset{1}{\to} i,i}}\varepsilon_{j\overset{1}{\to} i,i}
            +\frac{m(m-1)-\mb(\mb-1)}{2}M^{\kappa_{j\overset{2}{\to} i,i}}\varepsilon_{j\overset{2}{\to} i,i}\\
            &\quad\quad\quad+\frac{\mb(\mb+1)}{2}M^{\kappa_{j\overset{1}{\to} i,j}}\varepsilon_{j\overset{1}{\to} i,j}
            +\frac{m(m+1)-\mb(\mb+1)}{2}M^{\kappa_{j\overset{2}{\to} i,j}}\varepsilon_{j\overset{2}{\to} i,j}\\
            &\quad\quad\quad+0+0+0+0\Biggr)
            \times e^{\lambda\biggl(\kappa_{j\overset{1}{\to}i\overset{1}{\to}j}+\kappa_{j\overset{2}{\to}i\overset{2}{\to}j}\biggr)}\\
            =& \rho e^{\lambda m} + \Biggl(\frac{\mb(\mb-1)}{2}M^{\kappa_{j\overset{1}{\to} i,i}}\varepsilon_{j\overset{1}{\to} i,i}
            +\frac{m(m-1)-\mb(\mb-1)}{2}M^{\kappa_{j\overset{2}{\to} i,i}}\varepsilon_{j\overset{2}{\to} i,i}\\
            &\quad\quad\quad+\frac{\mb(\mb+1)}{2}M^{\kappa_{j\overset{1}{\to} i,j}}\varepsilon_{j\overset{1}{\to} i,j}
            +\frac{m(m+1)-\mb(\mb+1)}{2}M^{\kappa_{j\overset{2}{\to} i,j}}\varepsilon_{j\overset{2}{\to} i,j}\Biggr)\\
            &\qquad\qquad\times e^{\lambda\biggl(\kappa_{j\overset{1}{\to}i\overset{1}{\to}j}+\kappa_{j\overset{2}{\to}i\overset{2}{\to}j}\biggr)}.
        \end{align*}
        In view of condition \eqref{e:maincondn9}, the above expression is at most equal to \(1\). It follows by the assertion of Theorem \ref{t:mainres1} that the switched system \eqref{e:swsys} is GES under a non-periodic \(\sigma\in\Sw\).

         Now, suppose that (b) holds. We have a \(j\to i\) path, \(P_{G}^{j\to i}\), that satisfies conditions \eqref{e:maincondn10}-\eqref{e:maincondn11}. Recall that by construction of \(P_{G}^{i\to j}\), the choice \(\varepsilon_{i\to j,i} = \varepsilon_{i\to j,j} = 0\) holds. We apply Corollary \ref{cor:mainres2} to show that \eqref{e:swsys} is stabilizable in this setting.

        We have
        \begin{align*}
            \zeta_{j\to i,i} &= \abs{P_{G}^{j\to i}}\delta(m-1) + p(m-1) + qm = \kb_{j\to i,i},\\
            \zeta_{j\to i,j} &= \abs{P_{G}^{j\to i}}\delta(m-1) + pm + q(m-1) = \kb_{j\to i,j},\\
            \intertext{and}
            \zeta_{j\to i\to j} &= \biggl(\abs{P_{G}^{j\to i}}\delta + p + q\biggr)m = \kb_{j\to i\to j}.
        \end{align*}

        Consequently, the left-hand side of \eqref{e:maincondn7} becomes
        \begin{align*}
            \rho e^{\lambda m} + \Biggl(&\frac{m(m-1)}{2}M^{\kb_{j\to i,i}}\varepsilon_{j\to i,i}+\frac{m(m+1)}{2}M^{\kb_{j\to i,j}}\varepsilon_{j\to i,j}
            +0+0\Biggr)\times e^{\lambda\zeta_{j\to i\to j}}\\
            =\rho e^{\lambda m} + \Biggl(&\frac{m(m-1)}{2}M^{\kb_{j\to i,i}}\varepsilon_{j\to i,i}+\frac{m(m+1)}{2}M^{\kb_{j\to i,j}}\varepsilon_{j\to i,j}
            \Biggr)
            \times e^{\lambda\zeta_{j\to i\to j}}.
        \end{align*}
        In view of condition \eqref{e:maincondn11}, the above expression is at most equal to \(1\). It follows by the assertion of Corollary \ref{cor:mainres2} that there exists a periodic \(\sigma\in\Sw\) that ensures GES of \eqref{e:swsys}.
   \end{proof}

    For two pairs of a \(j\to i\) path and an \(i\to j\) path, \((P_{G}^{j\toone i},P_{G}^{i\toone j})\) and \((P_{G}^{j\totwo i},P_{G}^{i\totwo j})\) on \(G(\P,E(\P))\), we define
    \begin{align*}
        \chi_{i\toone j,i} &= \abs{P_{G}^{i\toone j}}\delta(\mb-1)+\abs{P_{G}^{i\totwo j}}\delta(m-\mb)+p(m-1)+qm,\nonumber\\
        \chi_{i\totwo j,i} &= \abs{P_{G}^{i\toone j}}\delta\mb + \abs{P_{G}^{i\totwo j}}\delta(m-\mb-1)+p(m-1)+qm,\nonumber\\
        \chi_{i\toone j,j} &= \abs{P_{G}^{i\toone j}}\delta(\mb-1) + \abs{P_{G}^{i\totwo j}}\delta(m-\mb)+ pm + q(m-1),\\
        \chi_{i\totwo j,j} &= \abs{P_{G}^{i\toone j}}\delta\mb + \abs{P_{G}^{i\totwo j}}\delta(m-\mb-1)+ pm + q(m-1),\nonumber\\
        \chi_{j\toone i\toone j} &= \biggl(\abs{P_{G}^{i\toone j}}\delta+p+q\biggr)\mb,\nonumber\\
        \chi_{j\totwo i\totwo j} &= \biggl(\abs{P_{G}^{i\totwo j}}\delta+p+q\biggr)(m-\mb),\nonumber
    \end{align*}
    and for \(P_{G}^{j\toone i} = P_{G}^{j\totwo i} = P_{G}^{j\to i}\) and \(P_{G}^{i\toone j} = P_{G}^{i\totwo j} = P_{G}^{i\to j}\), we define
   \begin{align*}
        \cb_{i\to j,i} &= \abs{P_{G}^{i\to j}}\delta(m-1) + p(m-1) + qm,\nonumber\\
        \cb_{i\to j,j} &= \abs{P_{G}^{i\to j}}\delta(m-1) + pm + q(m-1),\\
        \cb_{j\to i\to j} &= \biggl(\abs{P_{G}^{i\to j}}\delta+p+q\biggr)m.\nonumber
   \end{align*}

   \begin{corollary}
   \label{cor:mainres4}
        Let \(i,j\in\P\) satisfy Assumption \ref{a:stable_combi}, \((j,i)\in E(\P)\), \((i,j)\notin E(\P)\), and \(\lambda\) be an arbitrary positive number satisfying \eqref{e:maincondn1}. The following hold:\\
             (a) Suppose that \(G(\P,E(\P))\) admits two \(i\to j\) paths, \(P_{G}^{i\tor j}\), \(r=1,2\), such that there exist scalars \(\varepsilon_{i\tor j,i}\), \(\varepsilon_{i\tor j,j}\), \(r=1,2\), small enough, satisfying
                \begin{align}
                \label{e:maincondn12}
                    \norm{F_{i\tor j,i}^{p,\delta}}\leq\varepsilon_{i\tor j,i},\:\:
                    \norm{F_{i\tor j,j}^{q,\delta}}\leq\varepsilon_{i\tor j,j},\:\:r = 1,2,
                \end{align}
                and
                \begin{align}
                \label{e:maincondn13}
                    &\rho e^{\lambda m} + \Biggl(\frac{\mb(\mb+1)}{2}M^{\chi_{i\overset{1}{\to} j,i}}\varepsilon_{i\overset{1}{\to} j,i}
            +\frac{m(m+1)-\mb(\mb+1)}{2}M^{\chi_{i\overset{2}{\to} j,i}}\varepsilon_{i\overset{2}{\to} j,i}\nonumber\\
            &\quad\quad\quad+\frac{\mb(\mb+1)}{2}M^{\chi_{i\overset{1}{\to} j,j}}\varepsilon_{i\overset{1}{\to} j,j}
            +\frac{m(m+1)-\mb(\mb+1)}{2}M^{\chi_{i\overset{2}{\to} j,j}}\varepsilon_{i\overset{2}{\to} j,j}\Biggr)\nonumber\\
            &\qquad\qquad\times e^{\lambda\biggl(\chi_{j\overset{1}{\to}i\overset{1}{\to}j}+\chi_{j\overset{2}{\to}i\overset{2}{\to}j}\biggr)}\leq 1.
                \end{align}
                Then there exits a non-periodic switching signal \(\sigma\in\Sw\) under which the switched system \eqref{e:swsys} is GES.\\
             (b) Suppose that \(G(\P,E(\P))\) admits one \(i\to j\) path, \(P_{G}^{i\to j}\), such that there exist scalars \(\varepsilon_{i\to j,i}\), \(\varepsilon_{i\to j,j}\), small enough, satisfying
                \begin{align}
                \label{e:maincondn14}
                    \norm{F_{i\to j,i}^{p,\delta}}\leq\varepsilon_{i\to j,i}\:\:\text{and}\:\:\norm{F_{i\to j,j}^{q,\delta}}\leq\varepsilon_{i\to j,j},
                \end{align}
                and
                \begin{align}
                \label{e:maincondn15}
                     \rho e^{\lambda m} + \Biggl(&\frac{m(m+1)}{2}M^{\cb_{i\to j,i}}\varepsilon_{i\to j,i}+\frac{m(m+1)}{2}M^{\cb_{i\to j,j}}\varepsilon_{i\to j,j}\Biggr)
                     \times e^{\lambda\cb_{j\to i\to j}}\leq 1.
                \end{align}
                Then there exists a periodic switching signal \(\sigma\in\Sw\) under which the switched system \eqref{e:swsys} is GES.
   \end{corollary}

   \begin{proof}
        Since \((j,i)\in E(\P)\), the graph \(G(\P,E(\P))\) admits a \(j\to i\) path, \(P_{G}^{j\to i} = j,(j,i),i\), that satisfies \(\abs{P_{G}^{j\to i}} = 0\), \(F_{j\to i,i} = 0\) and \(F_{j\to i,j} = 0\).\footnote{Recall that \cite[p. 601]{Bernstein} for a \(\mathcal{M}\in\R^{d\times d}\), \(\norm{\mathcal{M}} = 0\) if and only if \(\mathcal{M} = 0\).} Corollary \ref{cor:mainres4} follows under a similar set of arguments employed for proving Corollary \ref{cor:mainres3}.
   \end{proof}

   \begin{corollary}
   \label{cor:mainres5}
        Let \(i,j\in\P\) satisfy Assumption \ref{a:stable_combi} and \(\lambda\) be an arbitrary positive number satisfying \eqref{e:maincondn1}. Suppose that \((i,j)\) and \((j,i)\in E(\P)\). Then the switched system \eqref{e:swsys} is stabilizable.
   \end{corollary}

   \begin{proof}
        Since \((i,j)\) and \((j,i)\in E(\P)\), the graph \(G(\P,E(\P))\) admits a pair of \(j\to i\) path and \(i\to j\) path, \((P_{G}^{j\to i},P_{G}^{i\to j})\), that satisfy \(\abs{P_{G}^{j\to i}} = 0\), \(F_{j\to i,i} = 0\), \(F_{j\to i,j} = 0\) and \(\abs{P_{G}^{i\to j}} = 0\), \(F_{i\to j,i} = 0\), \(F_{i\to j,j} = 0\). The assertion of Corollary \ref{cor:mainres5} follows at once from Corollary \ref{cor:mainres2}.
   \end{proof}

   \begin{rem}
   \label{rem:compa1}
   \rm{
        A vast body of the literature relies on state-dependent switching signals for the stabilization of \eqref{e:swsys}, see e.g., the recent works \cite{Heemels2017,Fiacchini2018} and the references therein. In contrast, \cite{Fiacchini2016} proposes (among others) purely time-dependent (in particular, periodic) switching signals that ensure GES of \eqref{e:swsys}. In general, determining the existence of a Schur stable combination, \(A_{\overline{i}}^{\overline{p}}A_{\overline{j}}^{\overline{q}}\) for some \(\overline{i},\overline{j}\in\P\) and \(\overline{p},\overline{q}\in\N\), or the absence of it, is not a numerically easy task. Indeed, the choice of \(\overline{p},\overline{q}\) to check are many. A necessary and sufficient condition for the existence of stabilizing periodic switching signals based on LMIs is also provided in \cite{Fiacchini2016}. In our work, the scalars \(p\), \(q\) have only a finite number of choices, and hence the number of operations needed to determine an \(\A\), if exists, is bounded. The proposed set of stabilizing switching signals in this paper is purely time-dependent and not restricted to periodic constructions. However, our stability conditions are only sufficient and their non-satisfaction does not imply the non-existence of a switching signal \(\sigma\in\Sw\) under which \eqref{e:swsys} is GES.
   }
   \end{rem}

   \begin{rem}
    \label{rem:compa2}
    \rm{
    	The usage of matrix inequalities in the context of stabilization of \eqref{e:swsys} is vast. For instance, stability of \eqref{e:swsys} under the so-called min switching signal is guaranteed if the subsystems matrices \(A_{\ell}\), \(\ell\in\P\), satisfy the Lyapunov-Metzler inequalities \cite{Geromel2006} or the S-procedure characterization \cite{Heemels2017}. These are BMIs and are numerically difficult to verify. In addition, the necessary and sufficient condition for the existence of a stabilizing periodic switching signal proposed in \cite{Fiacchini2016} involves satisfaction of LMI-based conditions. In contrast, our stability conditions rely on scalar inequalities involving upper bounds on the Euclidean norms of a set of commutators of certain products of the subsystems matrices, a set of scalars obtained from the properties of these matrices, and the rate of decay of a Schur stable matrix combination formed by any two of the subsystems matrices. These scalar inequalities are numerically easier to verify than the matrix inequalities based stability conditions.
    }
    \end{rem}

    \begin{rem}
    \label{rem:compa4}
    \rm{
        Stabilizability of the switched system \eqref{e:swsys} under pre-specified constraints on the set of admissible switching signals is addressed earlier in \cite{Fiacchini2018}. However, classes of state-dependent switching signals are considered, and geometric properties of certain sets are employed to provide necessary and sufficient conditions for recurrent stabilizability, which in turn is a sufficient condition for stabilizability. In contrast, we deal with purely time-dependent and not necessarily periodic switching signals, and present a set of stability conditions that involves scalar inequalities.
        }
    \end{rem}

    \begin{rem}
    \label{rem:compa3}
    \rm{
        Commutation relations between the subsystems matrices or certain products of these matrices have been employed to study stability of the switched system \eqref{e:swsys} under arbitrary switching \cite{Narendra1994,Agrachev2012}, minimum dwell time switching \cite{stab_min_dwell}, restrictions on admissible switches between the subsystems \cite{stab_cycle}, and restrictions on admissible minimum and maximum dwell times on the subsystems \cite{stab_min_max_dwell} earlier in the literature. Our work in the current paper differs from the earlier works in the following aspects:
        \begin{enumerate}[label = \arabic*), leftmargin = *]
            \item The works \cite{Narendra1994, Agrachev2012, stab_min_dwell} consider all subsystems to be stable, while the works \cite{stab_cycle,stab_min_max_dwell} assume the existence of at least one stable subsystem. In this paper we consider all subsystems to be unstable and combine the restrictions on admissible switches between the subsystems \cite{stab_cycle} and admissible dwell times on the subsystems \cite{stab_min_max_dwell}. Notice that even though we assume the existence of a Schur stable combination formed by two of the subsystems matrices, our results do not rely on unrestricted switches between them. Consequently, the results of \cite{stab_cycle,stab_min_max_dwell} do not extend to our setting with a Schur stable subsystem replaced by a Schur stable combination of subsystems in a straightforward manner.
            \item The work \cite{stab_min_max_dwell} allows some (but not all) systems in the family \eqref{e:family} to be unstable, and the proposed class of stabilizing switching signals does not allow consecutive activation of unstable subsystems. In contrast, the sets of stabilizing switching signals studied in this paper are not restricted to obey a maximum number of subsystems that can be activated between the two subsystems forming a Schur stable combination as long as a favourable path between them is chosen.
            \item The type of matrix commutators used in our analysis differs from the earlier works. Indeed, the works \cite{Narendra1994,Agrachev2012,stab_cycle} consider commutators between individual subsystems matrices, while the works \cite{stab_min_dwell,stab_min_max_dwell} use commutators between products of individual subsystems matrices. In contrast, we employ commutators of products of different subsystems matrices and products of individual subsystems matrices. The choice of the commutator under consideration is, however, not unique. The crux of the (matrix) commutation relations based stability analysis for switched systems lies in splitting matrix products into sums and applying combinatorial arguments on them, see \cite{Agrachev2012}, where this analysis technique was introduced. One may employ different choices of commutators to split a matrix product into sums, thereby leading to different sets of sufficient solutions to Problem \ref{prob:mainprob}.
        \end{enumerate}
        }
    \end{rem}

    \begin{rem}
    \label{rem:compa5}
    \rm{
         We explained in \S\ref{s:prob_stat} why Problem \ref{prob:mainprob} does not admit a trivial solution even when the family of systems \eqref{e:family} admits at least one stable subsystem. It is, therefore, only natural to ask if our results are also helpful to cater to the said setting. Let \(\P_{S}\) and \(\P_{U}\) denote the sets of indices of the stable and unstable subsystems in \eqref{e:family}, respectively, \(\P = \P_{S}\sqcup\P_{U}\). Fix \(i\in\P_{S}\), \(j=i\) and \(p,q\in\{\delta,\delta+1,\ldots,\Delta\}\). Clearly, the matrix \(\A\), defined in Assumption \ref{a:stable_combi} is Schur stable and Fact \ref{fact:m_defn} holds. Notice that \((i,i)\notin E(\P)\) as that contradicts the restriction on maximum dwell time on the subsystems. As a result, one always has to rely on favourable paths in the sense of Theorem \ref{t:mainres1} and Corollary \ref{cor:mainres2}, which may consist of either all stable or all unstable or both stable and unstable subsystems, for stabilizability of \eqref{e:swsys}. Our earlier work \cite{stab_min_max_dwell} on stability of \eqref{e:swsys} under pre-specified restrictions on admissible dwell times on the subsystems, deals with stabilizing switching signals that do not activate unstable subsystems consecutively. The technique to solve Problem \ref{prob:mainprob} described in this remark leads to a more general set of stabilizing switching signals in the sense that admissible switches between the subsystems are pre-specified and no restriction on the number of consecutive activation of unstable subsystems is imposed.
        }
    \end{rem}

    \begin{rem}
    \label{rem:compa6}
    \rm{
        Algorithmic design of switching signals that preserve stability of switched nonlinear systems in the presence of exogenous inputs, under pre-specified restrictions on admissible switches between the subsystems and admissible dwell times on the subsystems, is studied by employing multiple Lyapunov-like functions \cite{Branicky'98} and graph-theoretic tools in \cite{kun2018}. The proposed method assumes the existence of at least one stable subsystem, and involves constructing negative weight cycles on the underlying weighted directed graph of a switched system. The existence of these cycles depends on the existence of Lyapunov-like functions corresponding to the subsystems that satisfy certain conditions individually and among themselves. Given a family of systems, designing such functions is, in general, a numerically difficult problem. In contrast, in this paper we consider all subsystems to be unstable and do not rely on verifying if suitable Lyapunov-like functions exist for a given family of systems. Our stability conditions are, however, limited to the setting of linear subsystems.
    }
    \end{rem}
\section{Numerical experiment}
\label{s:numex}
    Consider a family of systems \eqref{e:family} with \(N = 4\). We generate the matrices \(A_{\ell}\in\R^{2\times 2}\), \(\ell\in\P\), by selecting elements from the interval \([-1,1]\) uniformly at random. It is ensured that all the matrices are unstable. The numerical values of \(A_{\ell}\), \(\ell\in\P\) along with their eigenvalues are furnished in Table \ref{tab:dataset1}. We have \(M=1.41\).
    \begin{table}[htbp]
	\centering
	\begin{tabular}{|c | c | c|}
		\hline
		\(\ell\) & \(A_{\ell}\) & eigenvalues of \(A_{\ell}\)\\
		\hline
		\(1\) & \(\pmat{0.796323 &  -0.9122466\\-0.7126696 &  0.1040671}\) & \(1.3276544,-0.4272643  \)\\
		\hline
		\(2\) & \(\pmat{0.9660338 & -0.972049\\-0.6582197 & -0.94077}\) & \(  1.2571386,-1.2318748\)\\
		\hline
		\(3\) & \(\pmat{-0.5085495 & -0.6519882\\-0.7370684 & -0.5013346}\) & \( -1.1981757,0.1882916  \)\\
		\hline
		\(4\) & \(\pmat{-0.990773 &  0.8742857\\0.780567 & 0.9401844}\) & \(-1.2959586,1.24537  \)\\
		\hline
	\end{tabular}
	\caption{Description of the subsystems matrices}\label{tab:dataset1}
	\end{table}

    Let \(E(\P) = \{(1,2),(2,1),(2,3),(3,2)\),\((3,4)\),\((4,1)\) and \(\delta = 2\), \(\Delta = 3\). Notice that unrestricted switches between the subsystems \(1\) and \(2\), and \(2\) and \(3\) are allowed. However, they do not form a stable combination in the sense of Assumption \ref{a:stable_combi}.

    We compute that Assumption \ref{a:stable_combi} holds with \(i=1\), \(j=3\) and \(p=q=2\). Indeed, the matrix \(\A = A_{1}^{2}A_{3}^{2} = \pmat{0.3379075 &  0.2444357\\
   0.0176692  & 0.061251 }\) has eigenvalues \(0.3527252\) and \(0.0464333\). We obtain \(m=1\) and \(\rho = 0.42\). Let \(\lambda = 0.0001\). It follows that \(\mb = 1\) and \(\rho e^{\lambda m} = 0.42 < 1\). However, both \((i,j)\) and \((j,i)\notin E(\P)\). So, our results are useful to cater to this setting.

    We construct a directed graph, \(G(\P,E(\P))\), as shown below.
    \begin{center}
            \scalebox{1}{
        \begin{tikzpicture}[every path/.style={>=latex},base node/.style={draw,circle}]
            \node[base node]            (a) at (1,1)  { \(2\) };
            \node[base node]            (b) at (1,-1)  { \(4\) };
            \node[base node]            (c) at (2,0) { \(3\) };
            \node[base node]            (d) at (0,0) { \(1\) };

	        \draw[->] (a) edge (c);
            \draw[->] (c) edge[bend right] (a);
	        \draw[->] (a) edge[bend right] (d);
            \draw[->] (d) edge (a);
            \draw[->] (b) edge (d);
            \draw[->] (c) edge (b);
        \end{tikzpicture}}
    \end{center}

    Fix
    \begin{align*}
        P_{G}^{j\toone i} &= 3,(3,2),2,(2,1),1,\\
        P_{G}^{j\totwo i} &= 3,(3,4),4,(4,1),1,\\
        \intertext{and}
        P_{G}^{i\toone j} &= P_{G}^{i\totwo j} = 1,(1,2),2,(2,3),3.
    \end{align*}
    We have \(\abs{P_{G}^{j\toone i}} = \abs{P_{G}^{j\totwo i}} = \abs{P_{G}^{i\toone j}} = \abs{P_{G}^{i\totwo j}} = 1\),
    \begin{align*}
        \norm{F_{j\toone i,i}^{p,\delta}} &= \norm{A_{1}^{2}A_{2}^{2}-A_{2}^{2}A_{1}^2} = 0.02,\\
        \norm{F_{j\totwo i,i}^{p,\delta}} &= \norm{A_{1}^{2}A_{4}^{2}-A_{4}^{2}A_{1}^2} = 0.05,\\
        \norm{F_{j\toone i,j}^{q,\delta}} &= \norm{A_{3}^{2}A_{2}^{2}-A_{2}^{2}A_{3}^2} = 0.04,\\
        \norm{F_{j\totwo i,j}^{q,\delta}} &= \norm{A_{3}^{2}A_{4}^{2}-A_{4}^{2}A_{3}^2} = 0.08,\\
        \norm{F_{i\toone j,i}^{p,\delta}} &= \norm{F_{i\totwo j,i}^{p,\delta}} = \norm{A_{1}^{2}A_{2}^{2}-A_{2}^{2}A_{1}^2} = 0.02,\\
        \norm{F_{i\toone j,j}^{q,\delta}} &= \norm{F_{i\totwo j,j}^{q,\delta}} = \norm{A_{3}^{2}A_{2}^{2}-A_{2}^{2}A_{3}^2} = 0.04,\\
        \xi_{j\tor i,i} = \xi_{j\tor i,j} &= \xi_{i\tor j,i} = \xi_{i\tor j,j} = 4,\:\:r=1,2,\\
        \intertext{and}
        \xi_{j\toone i\toone j} &= 8,\:\:\xi_{j\totwo i\totwo j} = 0.
    \end{align*}

    Consequently, the left-hand side of condition \eqref{e:maincondn4} is
    \begin{align*}
        &0.42 + \biggl(0+1\times 1.41^{4}\times (0.04+0.02+0.04)\biggr)e^{0.0008}\\
        =&\:\: 0.82 < 1.
    \end{align*}

    The assertion of Theorem \ref{t:mainres1} holds and the switched system \eqref{e:swsys} is GES under a switching signal \(\sigma\in\Sw\) that activates the sequence of subsystems \(3,2,1,2\), followed by \(s\)-many instances of the sequence of subsystems \(3,4,1,2\), \(s=1,2,3,\ldots\), repeatedly with dwell time \(2\) units of time on all subsystems. \(\sigma\) is illustrated in Figure \ref{fig:ex_fig3}. We generate \(100\) different initial conditions \(x_{0}\) from the interval \([-1,1]^{2}\) uniformly at random. The corresponding \((\norm{x(t)})_{t\in\N_{0}}\) for the switched system \eqref{e:swsys} with \(\sigma\) as described above, are plotted in Figure \ref{fig:ex_fig2}. GES of \eqref{e:swsys} is observed.
     \begin{figure}[htbp]
    \begin{center}
        \includegraphics[scale = 0.4]{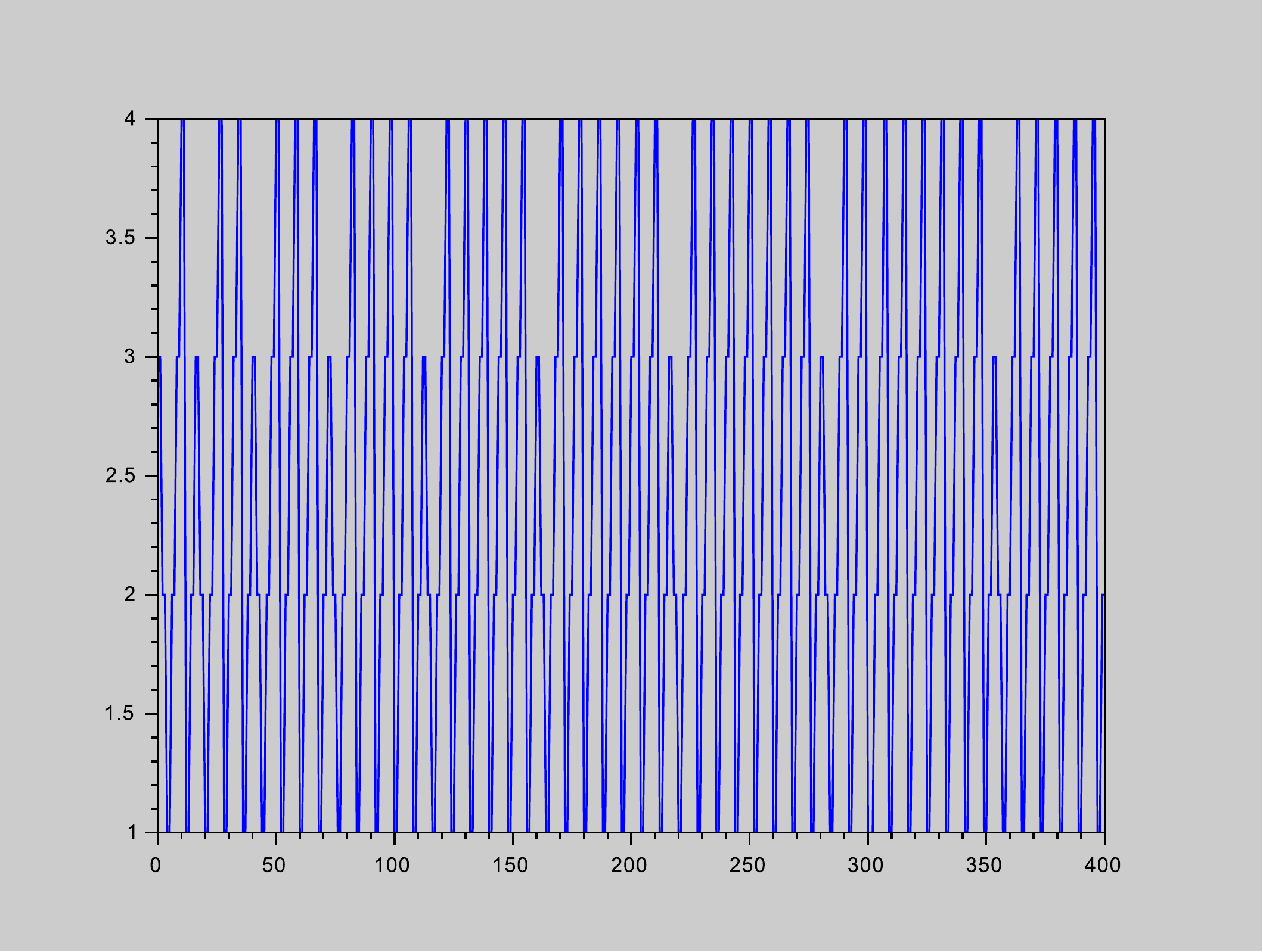}
        \caption{\((\sigma(t))_{t\in\N_{0}}\) obtained from our experiment.}\label{fig:ex_fig3}
    \end{center}
    \end{figure}

    \begin{figure}[htbp]
    \begin{center}
        \includegraphics[scale=0.4]{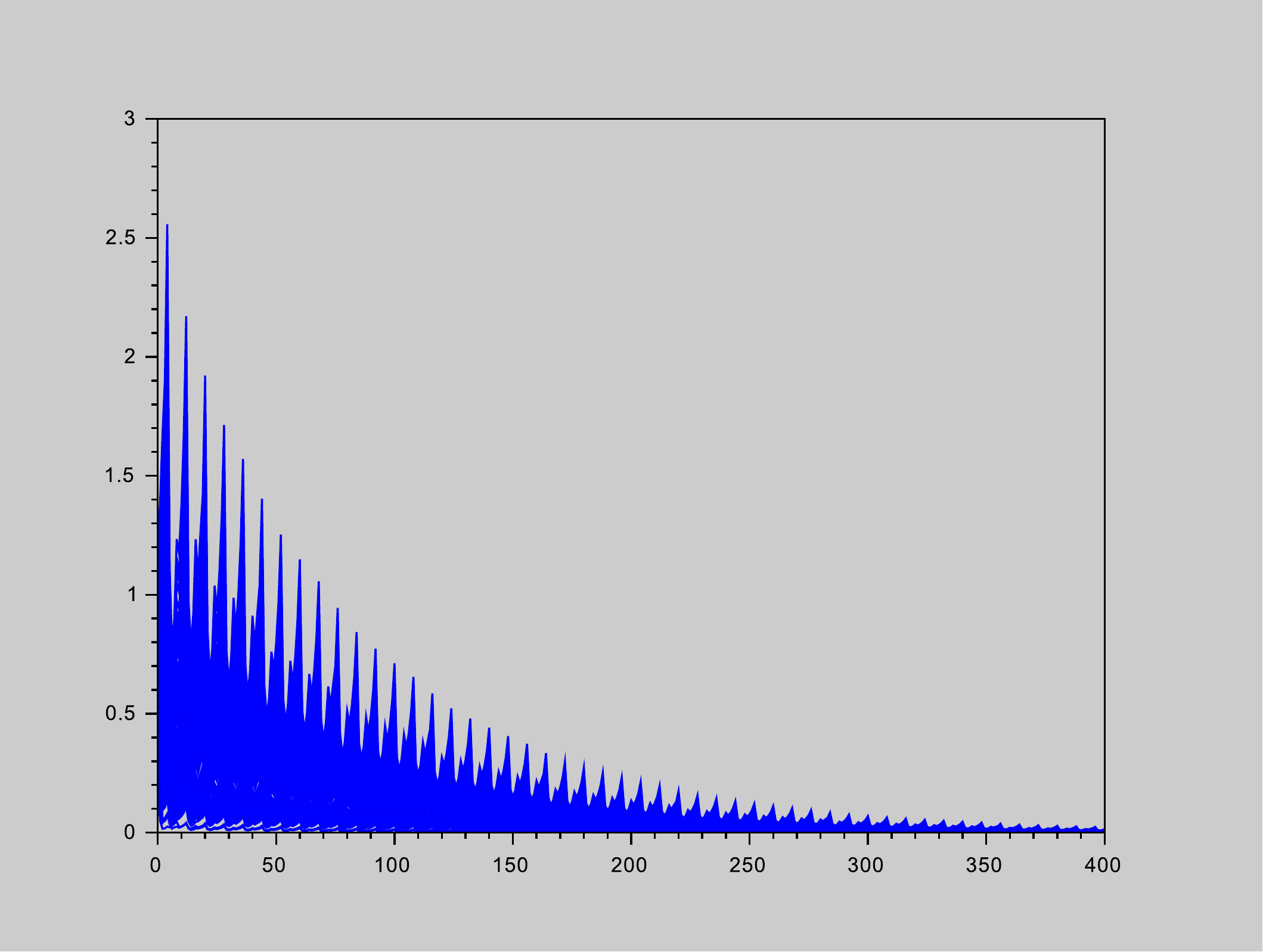}
        \caption{\((\norm{x(t)})_{t\in\N_{0}}\) for \eqref{e:swsys} under the \(\sigma\) in Figure \ref{fig:ex_fig3}.}\label{fig:ex_fig2}
    \end{center}
    \end{figure}

\section{Concluding remarks}
\label{s:concln}
    In this paper we studied stabilizability of discrete-time switched linear systems under restricted switching. Given restrictions on the set of admissible switches between the subsystems and admissible dwell times on the subsystems, we provided a set of sufficient conditions on the subsystems matrices such that there exist purely time-dependent switching signals that obey the given restrictions and ensure GES of the resulting switched system. Matrix commutators between certain products of the subsystems matrices and graph-theoretic arguments are employed as the main apparatuses for our analysis. Our stability conditions are derived in the premise of the existence of a suitable Schur stable combination formed by any two of the subsystems matrices. A next natural question is to address stabilizability of a switched system under purely time-dependent restricted switching signals when the subsystems matrices may not admit Schur stable combinations. This matter is currently under investigation and will be reported elsewhere.

\end{document}